\documentclass[runningheads]{llncs}
\usepackage[utf8]{inputenc}
\usepackage[T1]{fontenc}
\usepackage[english]{babel}
\usepackage{amssymb,amsmath}
\usepackage[font=small,center]{caption}
\usepackage{xcolor}
\usepackage{graphicx}
\usepackage{wrapfig}
\usepackage{makecell}
\usepackage{bussproofs}
\usepackage{fontawesome}
\usepackage{tabularx}
\usepackage{tikz}
\usetikzlibrary{positioning,arrows,shapes,hobby,backgrounds,calc,trees,fit}
\pgfdeclarelayer{background}
\pgfsetlayers{background,main}
\usepackage{indentfirst}
\makeatletter
\RequirePackage[bookmarks,unicode,colorlinks=true]{hyperref}%
   \def\@citecolor{blue}%
   \def\@urlcolor{blue}%
   \def\@linkcolor{blue}%

\def\orcidID#1{\smash{\href{http://orcid.org/#1}{\protect\raisebox{-1.25pt}{\protect\includegraphics{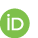}}}}}
\makeatother

\begin{document}
%
% ==========================================================
%
\title{A small-step approach to multi-trace checking against interactions (long version)} 
\titlerunning{Small-step multi-trace checking against interactions (long version)}
\author{Erwan Mahe\inst{1}\orcidID{0000-0002-5322-4337} \and
Boutheina Bannour \inst{2}\orcidID{0000-0002-4943-7807} \and
Christophe Gaston\inst{2}\orcidID{0000-0001-6865-5108} \and
Arnault Lapitre \inst{2}\orcidID{0000-0002-2185-4051} \and
Pascale Le Gall\inst{1}\orcidID{0000-0002-8955-6835} }

\authorrunning{E. Mahe, B. Bannour, C. Gaston, A. Lapitre, P. Le Gall}
\institute{Laboratoire de Mathématiques et Informatique pour la Complexité et les Systèmes\\
CentraleSupélec - Université de Paris-Saclay \\
9 rue Joliot-Curie, F-91192 Gif-sur-Yvette Cedex\and
CEA, LIST, Laboratory of Systems Requirements and Conformity Engineering, P.C. 174, Gif-sur-Yvette, 91191, France}
\maketitle
%
% ==========================================================
\definecolor{darkspringgreen}{rgb}{0.09, 0.45, 0.27}
% LF [22u8, 22u8, 130u8];
\definecolor{hibou_col_lf}{RGB}{22, 22, 130}
\newcommand{\hlf}[1]{\textcolor{hibou_col_lf}{#1}}
% MS [15u8, 86u8, 15u8];
\definecolor{hibou_col_ms}{RGB}{15, 86, 15}
\newcommand{\hms}[1]{\textcolor{hibou_col_ms}{#1}}
% NEW_FRESH [94u8, 22u8, 130u8]
\definecolor{hibou_col_newfresh}{RGB}{94, 22, 130}
\newcommand{\hfresh}[1]{\textcolor{hibou_col_newfresh}{#1}}
% ==========================================================
%
\begin{abstract}
Interaction models describe the exchange of messages between the different components of distributed systems. We have previously defined a small-step operational semantics for interaction models. The paper extends this work by presenting an approach for checking the validity of multi-traces against interaction models. A multi-trace is a collection of traces (sequences of emissions and receptions), each representing a local view of the same global execution of the distributed system. We have formally proven our approach, studied its complexity, and implemented it in a prototype tool. Finally, we discuss some observability issues when testing distributed systems via the analysis of multi-traces.
\keywords{interaction \and small-step operational semantics \and multi-trace analysis \and distributed system}
\end{abstract}
%
% ==========================================================
%
\section{Introduction}
{\em Context.} A distributed system (DS) can be viewed as a collection of sub-systems, which are distributed over distinct physical locations and which communicate with each other by exchanging messages~\cite{Lamport19b}. 
Analyzing the executions of DSs is a key problem to assess their correctness. However the 
distributed nature of observations complicates the investigation of bugs and undesirable behaviors. 
The absence of a global clock makes the classical notion of trace often too strong to represent DS executions. Indeed a trace fully orders all events occurring in it while ordering events occurring on remote locations is often impossible. Therefore, multi-traces are better suited to model executions of DSs. A multi-trace is a collection of traces, one per sub-system, which represents the sequence of actions - emissions and receptions of messages - that have been observed at its interface. 
Contrarily to traces, multi-traces do not strongly constrain orderings between actions occurring on different sub-systems.
Our work is related to the general problem of the automatic analysis and debugging of DSs based on local logging of traces~\cite{NevesMP18,BenharratGHLG17,Zaidi12,MaceRF15,AndresCN10}.
We are positioned at the intersection of two main issues: (1) that of tracking the causality of actions in traces \cite{NevesMP18,MaceRF15} based on the happened-before relation of Lamport \cite{Lamport19b} and (2) that of checking multi-traces against formal properties \cite{AndresCN10} or models \cite{Zaidi12,BenharratGHLG17}.

{\em Contribution.} In a model-based approach, we ground our analysis on models of interactions as the reference of intended DS executions. This kind of models - which include UML Sequence Diagrams~\cite{UML}, Message Sequence Charts~\cite{MSC}, BPMN Choregraphies~\cite{BPMN} among others - are widely used to specify DSs. In such models, DS executions are thought of as a coordination of message exchanges between multiple sub-systems. This allows for a specification from a global perspective. We consider interaction models where the execution units are actions (the same as those constituting traces) and can be combined using operators of sequencing, choice, repetition and parallelism. In a previous work~\cite{fase2020}, we have proposed a small-step operational semantics for interactions, backed by an equivalent algebraic denotational semantics. This paper presents an approach to check the validity of multi-traces against interaction models.
Validity refers to the notion of being an \emph{accepted} multi-trace, intuitively reflecting the fact of fully realizing \emph{one} of the behaviours prescribed by the reference interaction model, taking into account that interaction models can be non-deterministic. We prove the correctness and discuss the complexity class of our method for analyzing multi-traces w.r.t. interaction model semantics. Moreover, we discuss observability issues arising when testing distributed systems via the analysis of multi-traces.

As part of our contribution, we also developed a prototype implementing the small-step semantics and the multi-trace analysis. This tool is able to render graphical representations detailing the steps taken by the analysis. Images of interactions in this paper were adapted from its outputs.

{\em Related work.} Interaction models have been extensively used to validate DSs using Test Case generation~\cite{DanHierons11,BannourGS11,Longuet12}. 
Much effort is spent on the generation of local test cases to mitigate the following problems: (1) "observability" i.e. the difficulty in inferring global executions from partial visions of message exchanges
and (2) "controllability" i.e. the difficulty in determining when to apply stimuli in order to realize a targeted global execution. 
Our work, however, falls within another domain which is Passive Testing~\cite{AndresCN10,Zaidi12} (in which testers are only observers), and discusses other problems such as the Test Oracle Problem~\cite{Hierons14} (determining expected outputs w.r.t. given stimuli). 
Both works~\cite{AndresCN10,Zaidi12} have proposed approaches to check a set of local logs recorded in Service Oriented Systems. 
Authors in~\cite{AndresCN10} propose a methodology to verify the conservation of invariants during the execution of the system. Both local and global invariants can be checked although the latter is more costly in terms of computations.
%use a set of properties, called invariants. A centralized log of a web service is matched with a property in the form of a "body" which implies some "consequent" of visible actions (local invariant). Properties of web services can be combined (into a global invariant) and are used to check correlation of actions in logs. 
Our approach is different in that the reference for the analysis is not a logical property but a model of interaction as in~\cite{Zaidi12,Hierons14}. \cite{Zaidi12} discusses passive testing against choreography models expressed in the language Chor~\cite{Chor07}. It differs however from our approach in so far as: (1) Chor is less expressive than the interaction language we propose (particularly w.r.t. the absence of weak sequencing and the nature of loops), (2) \cite{Zaidi12} only handles synchronous communication between services, which cannot always describe accurately concrete implementations and (3) the local logs are not directly checked against the model but first pass through a synthesis step in which a global log is reconstituted and then this global log is checked. 
%considers logs analysis against interaction models expressed with the language Chor~\cite{Chor07}. Chor models can be structured with operator of sequence, exclusive choice and parallel. Those are endowed with concurrency semantics less complicated \commental{Peut-être ont-ils su la rendre plus simple ? ou sinon : less expressive ?} than interactions introduced~\cite{fase2020}. Moreover, communication is synchronous in Chor. This specific hypothesis is assumed equally on executions of web services which allows authors to synthesize a global log by merging the local logs and then check the resulting log against the model. In this paper, synchrony is not assumed that is an emission can be consumed at a later time, which covers a wide class of DS in the absence of a global clock. 
Authors in~\cite{Hierons14} investigate the computational cost of log analysis w.r.t. graphs of MSCs. This cost is compared in different cases according to the quality of observations (local or tester observability i.e. whether one have a set of independent local logs or a globally ordered log) and the expressivity of the MSC graphs (presence of choice, loop or parallelism).
%for different types of observability used in testing. 
The work echoes results for "MSC Membership"~\cite{AlurEY01,GenestM08} which state that this problem is NP-complete. The main factor of the cost blow-up lies in the fact that distributed actions can be equally re-ordered in multiple ways. 
%Algorithms have been proposed in~\cite{Hierons14} for logs analysis but no implementation is identified \commental{Un peu léger comme argument}. These algorithms consider restrictions on the language of MSCs in order to reduce analysis cost (bound on the observation probes, "safely realizable" when refined into automata products over bounded buffers with no implied scenarios and no deadlocks, ...). 
Our work is in the lineage of such research but we rather consider richer interaction models (asynchronous communications, weak sequencing, no enforced fork-join, ...). As such our language is closer to the appealing expressiveness of UML Sequence Diagrams. We therefore expect higher computational costs. Nevertheless, by applying a small-step semantics guided by the reading of the multi-trace, only pertinent parts of the search space are explored.
%In addition, we propose an algorithm which has been implemented as a software tool and proven via an automated theorem prover.

{\em Plan.} This paper is organized as follows. 
Sec.\ref{sec:overview} presents (multi-)traces and the concrete syntax of our interaction language.
Sec.\ref{sec:semantics} describes how interaction terms can be rewritten so as to define a small-step semantics in the form of accepted traces or multi-traces.
Sec.\ref{sec:analysis} presents our multi-trace analysis as well as some theoretical properties (termination, membership characterization, NP-hardness) and discuss a possible extension of our approach to take into account observability problems.
Finally Sec.\ref{sec:conclusion} concludes the paper.

\section{(Multi-)Traces and Interactions\label{sec:overview}}

Our goal is to analyse the validity of DS executions collected in the form of sets of local logs called multi-traces, w.r.t. a given interaction model. We now introduce the basic notions required to manipulate those concepts.

The description of a DS requires distinguishing between its distinct independent sub-systems and the different messages those sub-systems can exchange. In this paper, those sub-systems are abstracted as so-called {\em lifelines} (as in most Interaction-based languages) and we will note $L$ the set of all lifelines and $M$ the set of all messages. In the rest of the paper, $L$ and $M$ will be left implicit.

The basic building blocks of both (multi-)traces and interactions are actions. An action is either the emission or the reception of a message $m$ from or towards a lifeline $l$, noted respectively $l!m$ and $l?m$. We note the set of all actions with $Act=\{ l \Delta m ~|~ l \in L, ~ \Delta \in \{!,?\}, ~ m \in M\}$. When $L$ is reduced to a singleton $\{l\}$, we note $Act(l)$.
For an action $act$ of the form $l \Delta m$, $lf(act)$ denotes $l$.

\subsection{(Multi-)Traces}

A trace characterizes a given execution of a DS as a sequence\footnote{For a set $X$, $X^*$ denotes the set of sequences of elements of $X$ with $\epsilon$ being the empty sequence and the dot notation ($.$) being the concatenation law.} of actions (from $Act^*$), which appear in the order in which they occurred globally.

Given $L=\{l_1,\cdots,l_n\}$, a multi-trace is a tuple of traces $\mu=(\sigma_1,\cdots,\sigma_n)$ where, for any $j\in[1,n]$, $\sigma_j \in Act(l_j)^*$. A multi-trace therefore describes the execution of a DS as the collection of traces locally observed on each sub-system. 
Multi-traces do not constrain orderings between actions occurring on different lifelines. We note $Mult=$\resizebox{!}{8pt}{$\prod\limits_{l \in L}$}$Act(l)^*$ the set of multi-traces.

We may use the projection operator $proj$ from Def.\ref{def:projection_into_multi_trace} to project any trace $\varsigma \in Act^*$ into a multi-trace $proj(\varsigma) \in Mult$.

\begin{definition}[Trace Projection\label{def:projection_into_multi_trace}]
$proj : Act^* \rightarrow Mult$ is s.t.:
\begin{itemize}
    \item $proj(\epsilon) = ( \epsilon, \cdots, \epsilon )$
    \item given $j\in[1,n]$ and $act \in Act(l_j)$ and $\varsigma \in Act^*$\\ if $proj(\varsigma) = (\sigma_1,\cdots,\sigma_j,\cdots,\sigma_n)$ then $proj(act.\varsigma) = (\sigma_1,\cdots,act.\sigma_j,\cdots,\sigma_n)$.
\end{itemize}
\end{definition}

For instance, if we consider the trace $\varsigma = a!m_1.c?m_1.c!m_2.d?m_2$ defined over $M=\{m_1,m_2\}$ and $L = \{a,b,c,d\}$ then:
\[
\begin{array}{ccccc}
proj(\varsigma) = (
&
a!m_1,~~
&
\epsilon,~~
&
c?m_1.c!m_2,~~
&
d?m_2~)
\end{array}
\]

\subsection{Interaction Language}

\begin{figure} 
\vspace*{-.75cm}
\centering
\begin{tabular}{ccc}
\makecell{\resizebox{0.415\textwidth}{!}{
\begin{tikzpicture}[every node/.style = {shape=rectangle, align=center}]
\node (o) { $seq$ } [sibling distance=1.75cm,level distance=0.75cm]
  child {node (o1) {$loop_{seq}$}
    child {node (o11) {$seq$} [sibling distance=1.75cm]
      child {node (o111) {$strict$} [sibling distance=.85cm]
        child {node (o1111) {\hlf{$a$}$!$\hms{$m_1$}}}
        child {node (o1112) {\hlf{$b$}$?$\hms{$m_1$}}}
      }
      child {node (o112) {$seq$} [sibling distance=.85cm]
        child {node (o1121) {$alt$} [sibling distance=.85cm]
          child {node (o11211) {$strict$} [sibling distance=.85cm]
            child {node (o112111) {\hlf{$b$}$!$\hms{$m_2$}}}
            child {node (o112112) {\hlf{$c$}$?$\hms{$m_2$}}}
          }
          child {node (o11212) {$\varnothing$} }
        }
        child {node (o1122) {\hlf{$b$}$!$\hms{$m_3$}} }
      }
    }
  }
  child {node (o2) {$par$} [sibling distance=.85cm]
    child {node (o21) {\hlf{$a$}$!$\hms{$m_1$}}}
    child {node (o22) {$strict$} [sibling distance=.85cm]
      child {node (o221) {\hlf{$c$}$!$\hms{$m_4$}}}
      child {node (o222) {\hlf{$a$}$?$\hms{$m_4$}}}
    }
  }
;
\end{tikzpicture}
}}
&
\makecell{
\scriptsize $\leftarrow \left\{\begin{array}{c}
interaction\\term
\end{array} \right.$\\
\\\vspace*{1cm}
\scriptsize $\left. \begin{array}{c}
diagram\\repr.
\end{array} \right\} \rightarrow$
}
&
\makecell{\includegraphics[width=.325\textwidth]{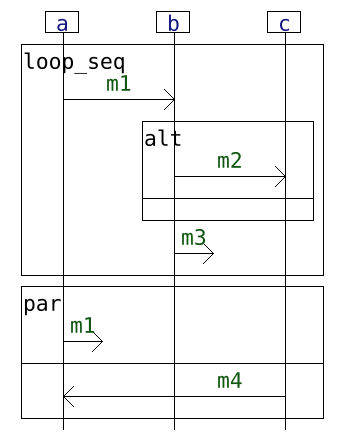}}
\end{tabular}
\caption{Example interaction}
\label{fig:ex3}
\vspace*{-.75cm}
\end{figure}
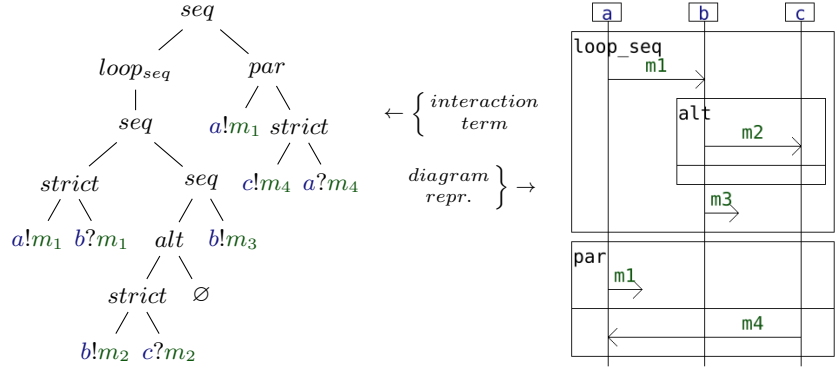

An interaction is a model which describes a DS by defining which are the actions that it may express and which are the possible orderings between those. 
As exemplified on the left of Fig.\ref{fig:ex3}, interactions are binary trees whose leaves are actions from $Act$. Precedence relations between $2$ actions at different leaf positions are then determined by the operators found in the inner nodes of the tree that separate those $2$ positions.

\begin{definition}[Interactions\label{def:labelled_interaction}]
The set $Int$ of {\em interactions} is s.t.:
\begin{itemize}
    \item $\varnothing \in Int$ and $Act \subset Int$,
    \item for $(i_{1},i_{2}) \in Int^{2}$ and $f \in \{strict,seq,alt,par\}$, $f(i_{1},i_{2}) \in Int$,
    \item for $ i \in Int$ and $f\in \{strict,seq,par\}$, $loop_f(i) \in Int$.
    \end{itemize}
\end{definition}

 The empty interaction $\varnothing$ and any action of $Act$ are basic interactions.
 $seq(i_1,i_2)$ (weak sequencing) indicates that actions specified by $i_1$ must occur before those of $i_2$ iff they occur on the same lifeline. 
 In contrast, $strict(i_1,i_2)$ (strict sequencing) imposes that actions specified by $i_1$ must occur before those of $i_2$ in any case.
 $par(i_1,i_2)$ allows actions from $i_1$ and $i_2$ to be fully interleaved while $alt(i_1,i_2)$ (exclusive alternative) specifies that either actions specified by $i_1$ or by $i_2$ occur. 
 As for the loop operators, $loop_f$ with $f \in \{seq,strict,par\}$, the index $f$ indicates with which binary operator loop unrollings have to be composed: in other words $loop_f(i_1)$ is equivalent to the term $alt(\varnothing,f(i_1,loop_f(i_1))$ (here we detailed the choice between not unrolling ($\varnothing$) and unrolling once).

Interactions can be illustrated by diagrams (cf. right part of Fig.\ref{fig:ex3}). Lifelines are depicted as vertical lines and actions $l\Delta m$ as arrows carrying their specific message $m$ and originating from or pointing towards their specific lifeline $l$. 
The passing of a message from a lifeline to another is modelled using the $strict$ operator (e.g. $strict(a!m,b?m)$ to denote the passing of $m$ from $a$ to $b$).
In diagrams, a message passing is depicted as an arrow from source to target lifeline.

\begin{wrapfigure}{l}{0.45\textwidth}
\vspace{-0.7cm}
\centering
\begin{tabular}{cc}
\makecell{\resizebox{0.225\textwidth}{!}{
\begin{tikzpicture}[every node/.style = {shape=rectangle, align=center}]
\node (o) {$seq$} [sibling distance=.85cm,level distance=0.75cm]
  child {node (o1) {$alt$} [sibling distance=.75cm]
    child {node (o11) {$strict$} [sibling distance=.85cm]
      child {node (o111) {\hlf{$b$}$!$\hms{$m_2$}}}
      child {node (o112) {\hlf{$c$}$?$\hms{$m_2$}}}
    }
    child {node (o12) {$\varnothing$}}
  }
  child {node (o21) {\hlf{$b$}$!$\hms{$m_3$}}};
\end{tikzpicture}
}}
&
\makecell{\includegraphics[width=.225\textwidth]{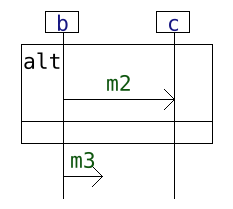}}
\end{tabular}
\caption{Small example}
\label{fig:subsubex3}
\vspace{-.75cm}
\end{wrapfigure}
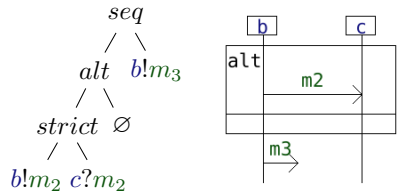

Let us consider the example from Fig.\ref{fig:subsubex3} (subterm of the one from Fig.\ref{fig:ex3}). Firstly, $b$ can either send $m_2$ to $c$ or not send anything. This choice is modelled by the $alt$ alternative operator. Secondly, $b$ must send $m_3$ to the environment. The implicit sequencing that we have described in natural language with the adverbs "firstly" and "secondly" is modelled by the $seq$ weak sequencing operator, which, unlike the other operators that are drawn explicitly with boxes, is implicitly represented by the top to bottom direction.

The semantics of an interaction $i$ can be defined as a set of global traces $Accept(i)$ or (not equivalently) by a set of of multi-traces $AccMult(i)$. Fig.\ref{fig:semantics_enumeration_ex} enumerates those semantics for the interaction from Fig.\ref{fig:subsubex3}. Let us note that an interleaving between $b!m_3$ and $c?m_2$ is noticeable in $Accept(i)$ but is not in $AccMult(i)$.

\begin{figure}
\vspace*{-.5cm}
    \centering
\begin{tabularx}{\textwidth}{|X|X|}
\hline
\makecell{
$
Accept(i) = 
\left\{
\begin{array}{l}
 \hlf{b}!\hms{m_2}.\hlf{c}?\hms{m_2}.\hlf{b}!\hms{m_3},\\
 \hlf{b}!\hms{m_2}.\hlf{b}!\hms{m_3}.\hlf{c}?\hms{m_2},\\
 \hlf{b}!\hms{m_3}
\end{array}
\right\}
$
}
&
\makecell{
$
AccMult(i) = 
\left\{
\begin{array}{l}
 (\hlf{b}!\hms{m_2}.\hlf{b}!\hms{m_3},~~\hlf{c}?\hms{m_2}),\\
 (\hlf{b}!\hms{m_3},~~\epsilon)
\end{array}
\right\}
$
}
\\
\hline
\end{tabularx}
    \caption{Enumerations of finite trace and multi-trace semantics for a simple example}
    \label{fig:semantics_enumeration_ex}
\vspace*{-.75cm}
\end{figure}

\section{Accepted (multi-)traces\label{sec:semantics}}

So as to formally define the set of accepted (multi-)traces of an interaction $i$, we reformulate semantic rules \resizebox{!}{10pt}{$i \xrightarrow{act} i'$} from \cite{fase2020} without relying on some denotational counterpart (in particular, without referring to notions of precedence relations between actions, as in \cite{KnappM17,fase2020}). 
To do this, in Sec.\ref{subsec:static}, we extract information statically from the term structure of interactions. 
This information is required to define, in Sec.\ref{subsec:execution}, the small-step interaction execution function $\chi$ grounding our operational approach.
Finally, in Sec.\ref{subsec:multitr_semantics}, we provide interactions with their two semantics: $Accept$, based on global traces, and $AccMult$, obtained by projection of $Accept$.

\subsection{Static analysis of interactions\label{subsec:static}}

As an interaction $i$ can contain several occurrences of the same action $act$, our small-steps do not correspond to transformations of the form \resizebox{!}{10pt}{$i \xrightarrow{act} i'$} bur rather \resizebox{!}{10pt}{$i \xrightarrow{act@p} i'$} where $p$ indicates the position of a specific occurrence of $act$ within $i$.

\begin{wrapfigure}{l}{0.35\textwidth}
\vspace{-0.7cm}
\centering
\begin{tikzpicture}[every node/.style = {shape=rectangle, align=center}]
\node (mo) {$seq$} [sibling distance=.85cm,level distance=0.75cm]
  child {node (mo1) {$alt$} [sibling distance=.75cm]
    child {node (mo11) {$strict$} [sibling distance=.85cm]
      child {node (mo111) {\hlf{$b$}$!$\hms{$m_2$}}}
      child {node (mo112) {\hlf{$c$}$?$\hms{$m_2$}}}
    }
    child {node (mo12) {$\varnothing$}}
  }
  child {node (mo2) {\hlf{$b$}$!$\hms{$m_3$}}};
\node[right=1.25cm of mo] (o) {$\epsilon$} [sibling distance=.85cm,level distance=0.75cm]
  child {node (o1) {$1$} [sibling distance=.75cm]
    child {node (o11) {$11$} [sibling distance=.85cm]
      child {node (o111) {$111$}}
      child {node (o112) {$112$}}
    }
    child {node (o12) {$12$}}
  }
  child {node (o2) {$2$}};
\end{tikzpicture}
\caption{Positions}
\label{fig:positions}
\vspace{-.5cm}
\end{wrapfigure}
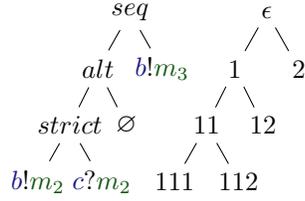

To do so, we use positions expressed in the Dewey Decimal Notation \cite{Dershowitz_rewrite_systems}. As the arity of our operators is at most 2, positions are defined as elements of $\{1,2\}^*$. A sub-interaction of an interaction $i$ at position $p$ is noted $i_{|p}$. Fig.\ref{fig:positions} illustrates positions within the interaction from Fig.\ref{fig:subsubex3}.

Moreover, for any set $P \in \mathcal{P}(\{1,2\}^*)$ and $x\in \{1,2\}$, we will note $x.P$ the set $\{ x.p ~|~ p \in P\}$.

The set $pos(i)$ of all well-defined positions w.r.t. an interaction $i$ is given below:

\begin{definition}[Positions\label{def:positions}]
$pos : Int \rightarrow \mathcal{P}(\{1,2\}^*)$ is defined as follows:
\begin{itemize}
\item $pos(\varnothing) = \{\epsilon\}$ and for any $act \in Act$, $pos(act) = \{\epsilon\}$
\item for any $i_1$ and $i_2$ in $Int$:
\begin{itemize}
    \item $pos(f(i_1,i_2)) = \{\epsilon\} \cup 1.pos(i_1) \cup 2.pos(i_2)$ with $f\in \{strict,seq,par,alt\}$
    \item $pos(loop_f(i_1)) = \{\epsilon\} \cup 1.pos(i_1)$ with $f\in \{strict,seq,par\}$
\end{itemize}
\end{itemize}
\end{definition}

We can then unambiguously designate sub-terms of an interaction $i$ (called sub-interactions) by using positions from $pos(i)$. For any $p \in pos(i)$, we use the notation $i_{|p}$ to refer to this sub-interaction of $i$ at position $p$. The formal definition of this notation is given below:

\begin{definition}[Sub-interactions\label{def:sub_interactions}]
$\__{|\_} : Int \times \{1,2\}^* \rightarrow Int$ is a partial function defined over couples $(i,p) \in Int \times \{1,2\}^*$ s.t. $p \in pos(i)$ as follows:
\begin{itemize}
\item $i_{|\epsilon} = i$ for any $i\in Int$
\item for any $i_1,i_2$ in $Int$ and $p_1\in pos(i_1)$ and $p_2 \in pos(i_2)$:
\begin{itemize}
    \item $(f(i_1,i_2))_{|1.p_1} = (i_1)_{|p_1}$ for any $f \in \{strict,seq,par,alt\}$
    \item $(f(i_1,i_2))_{|2.p_2} = (i_2)_{|p_2}$ for any $f \in \{strict,seq,par,alt\}$
    \item $(loop_f(i_1))_{|1.p_1} = (i_1)_{|p_1}$ for any $f \in \{strict,seq,par\}$
\end{itemize}
\end{itemize}
\end{definition}

The $exp_\epsilon$ function assesses statically whether or not an interaction accepts / expresses\footnote{we use the verb "express" in its name $exp_\epsilon$ so as not to risk confusion between this simple static function and the "accept" semantics defined later} the empty trace $\epsilon$. Naturally $\varnothing$ only accepts $\epsilon$, while interactions $act \in Act$ do not ($act$ must be executed). Similarly, any loop accepts $\epsilon$ because it is possible to repeat 0 times its content. The treatment of binary operators differs according to their intuitive meaning: for $alt$, it is sufficient that one of the two direct sub-interactions accepts $\epsilon$, while for the scheduling operators ($seq$, $strict$ and $par$), both have to accept $\epsilon$.

\begin{definition}[Emptiness\label{def:emptiness}]
$exp_{\epsilon} : Int \rightarrow bool$ is the function such that:
\begin{itemize}
\item $exp_{\epsilon}(\varnothing) = \top$ and for any $act \in Act$, $exp_{\epsilon}(act) = \bot$,
\item for any $i_1$ and $i_2$ in $Int$:
\begin{itemize}
    \item $exp_\epsilon(f(i_1,i_2)) = exp_{\epsilon}(i_1) \wedge exp_{\epsilon}(i_2)$ with $f\in \{strict,seq,par\}$
    \item $exp_\epsilon(alt(i_1,i_2)) = exp_{\epsilon}(i_1) \vee exp_{\epsilon}(i_2)$ 
    \item $exp_\epsilon(loop_f(i_1)) = \top$ with $f\in \{strict,seq,par\}$
\end{itemize}
\end{itemize}
\end{definition}

\begin{wrapfigure}{l}{0.35\textwidth}
\vspace{-1cm}
\centering
%\resizebox{0.35\textwidth}{!}{
\scalebox{1.15}{
\begin{tikzpicture}[every node/.style = {shape=rectangle, align=center}]
\node (o) {$seq${\scriptsize \textcolor{red}{\faTimesCircle}} } [sibling distance=1.25cm,level distance=0.75cm]
  child {node (o1) {$alt${\scriptsize \textcolor{darkspringgreen}{\faCheckCircle}}} [sibling distance=1.25cm]
    child {node (o11) {$strict${\scriptsize \textcolor{red}{\faTimesCircle}}} [sibling distance=1.25cm]
      child {node (o111) {\hlf{$b$}$!$\hms{$m_2$}{\scriptsize \textcolor{red}{\faTimesCircle}}}}
      child {node (o112) {\hlf{$c$}$?$\hms{$m_2$}{\scriptsize \textcolor{red}{\faTimesCircle}} }}
    }
    child {node (o12){ $\varnothing${\scriptsize \textcolor{darkspringgreen}{\faCheckCircle}}}}
  }
  child {node (o2) {\hlf{$b$}$!$\hms{$m_3$}{\scriptsize \textcolor{red}{\faTimesCircle}} }};
\end{tikzpicture}
}
\caption{$exp_\epsilon$}
\label{fig:express_empty_ex}
\vspace{-.75cm}
\end{wrapfigure}
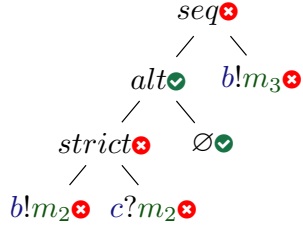

Let us consider Fig.\ref{fig:express_empty_ex}. We can recognize the interaction term of the example from Fig.\ref{fig:subsubex3}. This interaction does not express the empty trace i.e. $exp_\epsilon(i) = \bot$. Fig.\ref{fig:express_empty_ex} illustrates how $exp_\epsilon(i)$ is computed. We start from the leaf nodes. 

On the leaf node that is the empty interaction at position $12$ (i.e. $i_{|12} = \varnothing$) we have immediately that $exp_\epsilon(i_{|12}) = \top$. Indeed, the empty interaction describes an execution where no action is expressed, leading to the empty trace $\epsilon$.
We note that on Fig.\ref{fig:express_empty_ex} by drawing a {\scriptsize\textcolor{darkspringgreen}{\faCheckCircle}} symbol on the immediate right of the root node of this sub-interaction $i_{|12}$.

On the leaf nodes hosting actions, at positions $111$, $112$ and $2$, we immediately have $exp_\epsilon(i_{|111})= exp_\epsilon(i_{|112})= exp_\epsilon(i_{|2}) = \bot$.
Indeed, an interaction that is reduced to a single action $act$ describes the mandatory execution of that action, leading to a trace $\varsigma = act$ of length $1$, which is not the empty trace.
We then decorate on Fig.\ref{fig:express_empty_ex} those respective nodes with the {\scriptsize\textcolor{red}{\faTimesCircle}} symbol, signifying that those sub-interactions do not express the empty trace.

Let us then consider sub-interaction $i_{|11} = strict(b!m_2,c?m_2)$. Given that we have computed $exp_\epsilon$ for its child sub-interactions $i_{|111}$ and $i_{|112}$, we can immediately infer that $exp_\epsilon(i_{|11}) = \bot$. Indeed, the $strict$ operator is a scheduling operator i.e. it describes executions which are specific interleavings between $2$ executions: one occurring on the left sub-interaction (here $i_{|111}$), and one on the right (here $i_{|112}$). Hence, if, on at least one child sub-interaction, every execution expresses at least one action, then any execution of the parent interaction also expresses at least one action. Therefore $exp_\epsilon(i_{|11}) = \bot$. On Fig.\ref{fig:express_empty_ex}, we then decorate the root note of $i_{|11}$ (which is the node at position $11$) with the {\scriptsize\textcolor{red}{\faTimesCircle}} symbol.

Let us then consider sub-interaction $i_{|1} = alt(i_{|11},\varnothing)$. By definition, $i_{|1}$ accepts the empty trace. Indeed, given that its root node is an $alt$ operator, $i_{|1}$ describes an alternative between the expression of $2$ distinct behaviors, each one modelled by one of the two sub-interaction $i_{|11}$ and $i_{|12} = \varnothing$. $i_{|11}$ do not express the empty trace (in facts, it describes a single possible execution which leads to a trace $b!m_2.c?m_2$). $i_{|12}$ describes the empty execution, leading to the empty trace $\epsilon$. Therefore $i_{|1}$ can either produce trace $b!m_2.c?m_2$ or the empty trace. Hence we have $exp_\epsilon(i_{|1}) = \top$, which we note by symbol {\scriptsize\textcolor{darkspringgreen}{\faCheckCircle}} on the node at position $1$ on Fig.\ref{fig:express_empty_ex}.

Finally, whether or not the overall interaction $i=seq(i_{|1},i_{|2})$ expresses $\epsilon$ is determined as for any other sub-interaction. Here the root node is a scheduling $seq$ operator therefore so as to express $\epsilon$, $i$ would require both $i_{|1}$ and $i_{|2}$ to do so. As a result we have $exp_\epsilon(i) = \bot$.

The $avoids$ function states, for an interaction $i$ and a lifeline $l$, whether or not $i$ accepts an execution that involves no actions occurring on $l$.
The empty interaction $\varnothing$ "avoids" every lifeline. An action $l'\Delta m$ "avoids" $l$ iff it occurs on a different lifeline. Then, as for $exp_\epsilon$, $avoids$ is defined inductively. Any loop may avoid any lifeline given that, in any case, it is possible to repeat 0 times its content. For an interaction of the form $i=alt(i_{|1},i_{|2})$, it is sufficient that any one of the two sub-interactions $i_{|1}$ or $i_{|2}$ avoids $l$ so that $i$ may avoid $l$. For the scheduling operators ($seq$, $strict$ and $par$), both have to avoid $l$.

\begin{definition}[Avoiding\label{def:avoiding}]
We define the functions $avoids : Int \times L \rightarrow bool$ s.t. for any $l \in L$:
\begin{itemize}
\item $avoids(\varnothing,l) = \top$
\item $avoids(l'\Delta m,l) = (l' \neq l)$, for any $act = l'\Delta m \in Act$,
\item for any $i_1$ and $i_2$ in $Int$:
\begin{itemize}
    \item $avoids(f(i_1,i_2),l) = avoids(i_1,l) \wedge avoids(i_2,l)$ with $f\in \{strict,seq,par\}$
    \item $avoids(alt(i_1,i_2),l) = avoids(i_1,l) \vee avoids(i_2,l)$ 
    \item $avoids(loop_f(i_1),l) = \top$ with $f\in \{strict,seq,par\}$
\end{itemize}
\end{itemize}
\end{definition}

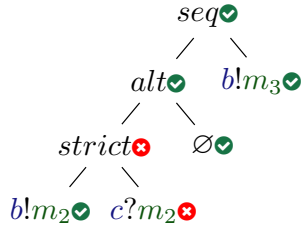
\begin{wrapfigure}{l}{0.35\textwidth}
\vspace{-1cm}
\centering
%\resizebox{0.35\textwidth}{!}{
\scalebox{1.15}{
\begin{tikzpicture}[every node/.style = {shape=rectangle, align=center}]
\node (o) {$seq${\scriptsize \textcolor{darkspringgreen}{\faCheckCircle}} } [sibling distance=1.25cm,level distance=0.75cm]
  child {node (o1) {$alt${\scriptsize \textcolor{darkspringgreen}{\faCheckCircle}}} [sibling distance=1.25cm]
    child {node (o11) {$strict${\scriptsize \textcolor{red}{\faTimesCircle}}} [sibling distance=1.25cm]
      child {node (o111) {\hlf{$b$}$!$\hms{$m_2$}{\scriptsize \textcolor{darkspringgreen}{\faCheckCircle}}}}
      child {node (o112) {\hlf{$c$}$?$\hms{$m_2$}{\scriptsize \textcolor{red}{\faTimesCircle}} }}
    }
    child {node (o12){ $\varnothing${\scriptsize \textcolor{darkspringgreen}{\faCheckCircle}}}}
  }
  child {node (o2) {\hlf{$b$}$!$\hms{$m_3$}{\scriptsize \textcolor{darkspringgreen}{\faCheckCircle}} }};
\end{tikzpicture}
}
\caption{$avoids(i,c)$}
\label{fig:avoid_ex}
\vspace{-.75cm}
\end{wrapfigure}

Fig.\ref{fig:avoid_ex} describes the application of $avoids(\_,c)$ on an interaction $i$ (from Fig.\ref{fig:subsubex3}) and its sub-interactions.
At the leaf nodes, $avoids(\_,c)$ is either $\top$ (on $\varnothing$ or on actions which do not occur on $c$) or $\bot$ (on actions which occur on $c$). In this interaction $i$, the only action occurring on lifeline $c$ is $c?m_2$. At this leaf node we therefore put the {\scriptsize \textcolor{red}{\faTimesCircle}} symbol to signify that $avoids(c?m_2,c) = \bot$. In all other leaf actions, we put the {\scriptsize \textcolor{darkspringgreen}{\faCheckCircle}} symbol. Then, $avoids(\_,c)$ is computed from bottom to top w.r.t. the interaction term, in the exact same manner as $exp_\epsilon$ would be. The value of $avoids(\_,c)$ on a parent interaction is inferred from the values computed on child sub-interactions depending on the nature of the parent operator.

In the following, for any $i,l\in Int\times L$ we will simply note:
\begin{itemize}
    \item $exp_\epsilon(i)$ when $exp_\epsilon(i) = \top$
    \item $\neg exp_\epsilon(i)$ when $exp_\epsilon(i) = \bot$
    \item $avoids(i,l)$ when $avoids(i,l) = \top$ 
    \item $\neg avoids(i,l)$ when $avoids(i,l) = \bot$.
\end{itemize}

%The $exp_\epsilon$ function behaves the same way, except that it returns $\bot$ on any action leave. In facts we have for any interaction $i$, $exp_\epsilon(i) = \bigwedge_{l\in L} avoid(i,l)$, meaning an interaction can express the empty trace iff it accepts an execution that can avoid all lifelines.

Among all actions leaves of $i$, only some are immediately executable.
The function $front$ (for {\em frontier}) in Def.\ref{def:frontier_labelled}, determines the positions of all such actions.

\begin{definition}[Frontier\label{def:frontier_labelled}]
$front: Int \rightarrow \mathcal{P}( \{1,2\}^* )$ is the function s.t.:
\begin{itemize}
\item  $front(\varnothing)=\emptyset$  and for any $act \in Act$, $front(act) = \{ \epsilon \}$,
\item for any $i_1$ and $i_2$ in $Int$:
    \begin{itemize}
        \item $front(strict(i_1,i_2)) = \left\{ \begin{array}{ll}
1.front(i_1) \cup 2.front(i_2)  & \text{if } exp_\epsilon(i_1) \\
1.front(i_1)  & \text{else}
\end{array} \right.$
        \item $front(seq(i_1,i_2)) = 1.front(i_1)\cup \{p ~|~ p \in 2.front(i_2) \; , \;  avoids(i_1,lf(i_{|p}) \}$,
        \item $front(f(i_1,i_2)) = 1.front(i_1) \cup 2.front(i_2)$ with $f \in \{ alt , par \}$
        \item $front(loop_f(i_1)) = 1.front(i_1)$ for $f \in \{strict,seq,par\}$.
    \end{itemize}
\end{itemize}
For any $p\in front(i)$, $i_{|p}$ is a called a frontier action.
\end{definition}

The empty interaction has an empty frontier: $front(\varnothing)=\emptyset$. For any action $act$, $front(act)=\{\epsilon\}$ ($\epsilon$ is the position of $act$ which is immediately executable).
For $i$ of the form $f(i_1,i_2)$, $front(i)$ is inferred from $front(i_1)$ and $front(i_2)$. In all cases, actions occurring at positions in $front(i_1)$ are immediately executable in $i$. Indeed, the term being read from left to right, all operators, if they introduce ordering constraints, will only do so on the right sub-interaction $i_2$. Thus $1.front(i_1)$ is included in $front(i)$. If $f=alt$ or $f=par$, $2.front(i_2)$ is also included in $front(i)$ because no constraint may prevent the execution of actions from $i_2$. If $f=strict$, any action from $i_2$ can only be executed if no action from $i_1$ is (otherwise it would violate the strict sequencing). Therefore $2.front(i_2)$ is included in $front(i)$ iff $i_1$ accepts the empty trace. If $f=seq$, elements $p$ from $2.front(i_2)$ are included in $front(i)$ iff $i_1$ accepts an execution that does not involve the lifeline on which the action $i_{|p}$ occurs.

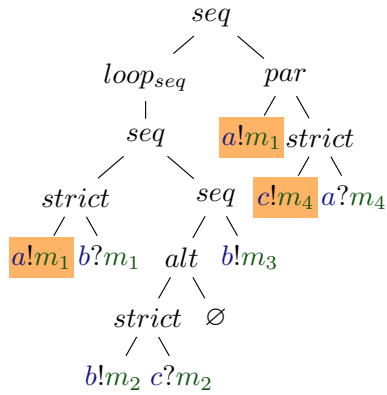
\begin{wrapfigure}{l}{0.425\textwidth}
\vspace*{-0.8cm}
    \centering
\resizebox{0.45\textwidth}{!}{
\begin{tikzpicture}[every node/.style = {shape=rectangle, align=center}]
\node (o) { $seq$ } [sibling distance=1.75cm,level distance=0.75cm]
  child {node (o1) {$loop_{seq}$}
    child {node (o11) {$seq$} [sibling distance=1.75cm]
      child {node (o111) {$strict$} [sibling distance=.85cm]
        child {node (o1111) {\hlf{$a$}$!$\hms{$m_1$}}}
        child {node (o1112) {\hlf{$b$}$?$\hms{$m_1$}}}
      }
      child {node (o112) {$seq$} [sibling distance=.85cm]
        child {node (o1121) {$alt$} [sibling distance=.85cm]
          child {node (o11211) {$strict$} [sibling distance=.85cm]
            child {node (o112111) {\hlf{$b$}$!$\hms{$m_2$}}}
            child {node (o112112) {\hlf{$c$}$?$\hms{$m_2$}}}
          }
          child {node (o11212) {$\varnothing$} }
        }
        child {node (o1122) {\hlf{$b$}$!$\hms{$m_3$}} }
      }
    }
  }
  child {node (o2) {$par$} [sibling distance=.85cm]
    child {node (o21) {\hlf{$a$}$!$\hms{$m_1$}}}
    child {node (o22) {$strict$} [sibling distance=.85cm]
      child {node (o221) {\hlf{$c$}$!$\hms{$m_4$}}}
      child {node (o222) {\hlf{$a$}$?$\hms{$m_4$}}}
    }
  }
;
\begin{pgfonlayer}{background}
\fill[orange!60] ($(o1111) + (-.4,-.25)$) -- ($(o1111) + (-.4,+.25)$) -- ($(o1111) + (+.4,+.25)$) -- ($(o1111) + (+.4,-.25)$) -- cycle;
\fill[orange!60] ($(o21) + (-.4,-.25)$) -- ($(o21) + (-.4,+.25)$) -- ($(o21) + (+.4,+.25)$) -- ($(o21) + (+.4,-.25)$) -- cycle;
\fill[orange!60] ($(o221) + (-.4,-.25)$) -- ($(o221) + (-.4,+.25)$) -- ($(o221) + (+.4,+.25)$) -- ($(o221) + (+.4,-.25)$) -- cycle;
\end{pgfonlayer};
\end{tikzpicture}
}
    \caption{Frontier actions (highlighted)}
    \label{fig:frontier_ex}
\vspace*{-1cm}
\end{wrapfigure}

Fig.\ref{fig:frontier_ex} illustrates the definition of $front$ on the example from Fig.\ref{fig:ex3}. Given the $8$ different actions on leaves, we have $front(i) \subseteq \{1111,1112,112111,112112,$ $1122,21,221,222\}$. Actions on the right of every $strict$ operators are prevented from being executed by those on their left and as such are not in the frontier. This eliminates $\{1112,112112,222\}$. $b!m_2$ and $b!m_3$ are prevented from being executed by $b?m_1$ which is a cousin on their left w.r.t the $seq$ operator at position $11$. This eliminates $\{1112,1122\}$. Then, by elimination, $front(i) = \{1111,21,221\}$.

\subsection{Interaction Execution\label{subsec:execution}}

We now define the small-step used in our operational semantics. It consists in transforming an interaction $i$ having the position $p$ in its frontier into an interaction $i'$ s.t. $i'$ characterizes in intentions all the possible futures of the execution of the action $i_{|p}$ according to $i$.

At first, we define a function that associates to any interaction $i$ that may avoid $l$ (i.e. s.t. $avoids(i,l)$), a new interaction, which characterizes exactly all the executions of $i$ that do not involve lifeline $l$.

\begin{definition}[Pruning\label{def:prune}]
The function $prune : Int \times L \rightarrow Int$ is defined for couples $(i,l)$ in $Int \times L$ verifying $avoids(i,l)$ by:
\begin{itemize}
    \item $prune(\varnothing,l) = \varnothing$ and for any $act \in Act$, $prune(act,l)=act$
    \item for any $(i_1,i_2) \in Int^2$, $prune(alt(i_1,i_2),l)$ is equal to:
    \begin{itemize}
        \item $prune(i_2,l)$ if $avoids(i_2,l) \wedge \neg avoids(i_1,l)$
        \item $prune(i_1,l)$ if $avoids(i_1,l) \wedge \neg avoids(i_2,l)$
        \item $alt(prune(i_1,l),prune(i_2,l))$ if $avoids(i_1,l) \wedge avoids(i_2,l)$
    \end{itemize}
    \item for any $(i_1,i_2) \in Int^2$ and any $f \in \{strict,seq,par\}$:
    \begin{itemize}
        \item $prune(f(i_1,i_2),l)= f(prune(i_1,l),prune(i_2,l))$
    \end{itemize}
    \item for any $i \in Int$ and any $f \in \{strict,seq,par\}$:
    \begin{itemize}
        \item $prune(loop_f(i),l)=loop_f(prune(i,l))$ if $avoids(i,l)$
        \item $prune(loop_f(i),l)=\varnothing$ if $\neg avoids(i,l)$
    \end{itemize}
\end{itemize}
\end{definition}

For any given lifeline $l$, $prune(\_,l) : Int \rightarrow Int$ eliminates from a given interaction $i$ (s.t. the precondition $avoids(i,l)$ is satisfied) all actions occurring on lifeline $l$ while preserving a maxima the original semantics of $i$ i.e. so that $Accept(prune(i,l)) \subseteq Accept(i)$ and $Accept(prune(i,l))$ is the maximum subset of $Accept(i)$ that contains no trace in which there are actions occurring on $l$.

So as to preserve the semantics, the interaction term $i$ can only be modified in two manners with the aim to eliminate actions: (1) by forcing the choice of a given sub-interaction in $alt$ nodes (illustrated on Fig.\ref{fig:pruning_of_alt}) and (2) by choosing to forbid the repetition of a sub-interaction in $loop$ nodes (illustrated on Fig.\ref{fig:pruning_of_loop}). Those modifications strictly correspond to the elimination of some possible executions of $i$ and therefore we have $Accept(prune(i,l)) \subseteq Accept(i)$.

We describe in the following the mechanism of $prune$ formalised in Def.\ref{def:prune}. Let's consider a lifeline $l$. We have $prune(\varnothing,l) = \varnothing$ because there is nothing to eliminate. For any action $act \in Act$, $prune(act,l)$ is well defined iff $avoids(act,l)$. Therefore, $act$ is not an action that needs to be eliminated and $prune(act,l) = act$. For $i=alt(i_{|1},i_{|2})$, in order for the precondition $avoids(i,l)$ to be satisfied, we have either or both of $avoids(i_{|1},l)$ or  $avoids(i_{|2},l)$. If both branches avoid $l$ they can be pruned and kept in the interaction term. If only a single one does, we only keep the pruned version of this single branch. For any scheduling operator $f$, if $i=f(i_{|1},i_{|2})$, in order to have $avoids(i,l)$ we must have both $avoids(i_{|1},l)$ and  $avoids(i_{|2},l)$. Then $prune(i,l)$ is simply defined as the scheduling by $f$ of the pruned versions of $i_{|1}$ and $i_{|2}$. Finally, for loops, i.e. with $i$ of the form $loop_f(i_{|1})$ with $f$ a scheduling operator, we distinguish two cases. (1) If $\neg avoids(i_{|1},l)$ then any execution of $i_{|1}$ will yield a trace containing actions occurring on $l$. Therefore it is necessary to forbid the repetition of the loop. This is done by specifying that $prune(i,l) = \varnothing$. (2) If $avoids(i_{|1},l)$ then it is not necessary to forbid the repetition of the loop, given that sub-interaction $i_{|1}$ can be pruned and therefore may not yield traces with actions occurring on $l$. This being the modification which preserves a maximum amount of traces of the semantics, we have $prune(i) = loop_f(prune(i_{|1},l))$. The recursive nature of $prune$ then guarantees that only the minimally required modifications are done on the interaction term so as to eliminate from it undesired actions.

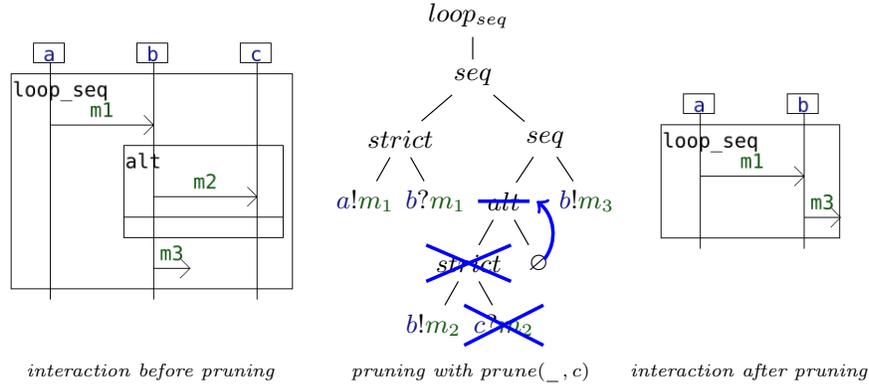
\begin{figure}
\vspace{-.75cm}
    \centering
\begin{tabular}{ccc}
\makecell{
\includegraphics[scale=.35]{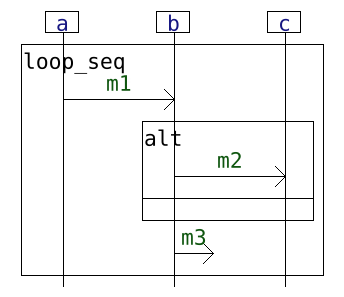}
}
&
\makecell{
\scalebox{1.1}{
\begin{tikzpicture}[every node/.style = {shape=rectangle, align=center}]
\node (o1) { $loop_{seq}$ } [sibling distance=1.75cm,level distance=0.75cm]
child {node (o11) {$seq$} [sibling distance=1.75cm]
      child {node (o111) {$strict$} [sibling distance=.85cm]
        child {node (o1111) {\hlf{$a$}$!$\hms{$m_1$}}}
        child {node (o1112) {\hlf{$b$}$?$\hms{$m_1$}}}
      }
      child {node (o112) {$seq$} [sibling distance=1cm]
        child {node (o1121) {$alt$} [sibling distance=.85cm]
          child {node (o11211) {$strict$} [sibling distance=.85cm]
            child {node (o112111) {\hlf{$b$}$!$\hms{$m_2$}}}
            child {node (o112112) {\hlf{$c$}$?$\hms{$m_2$}}}
          }
          child {node (o11212) {$\varnothing$} }
        }
        child {node (o1122) {\hlf{$b$}$!$\hms{$m_3$}} }
      }
    }
;
\draw[blue,very thick] (o112112.south east) -- (o112112.north west);
\draw[blue,very thick] (o112112.south west) -- (o112112.north east);
\draw[blue,very thick] (o11211.south east) -- (o11211.north west);
\draw[blue,very thick] (o11211.south west) -- (o11211.north east);
\draw[blue,very thick] (o1121.east) -- (o1121.west);
\draw[->,blue,very thick] ([xshift=-5pt,yshift=-5pt] o11212.north east) to [bend right=45] ([xshift=2.5pt] o1121.east);
\end{tikzpicture}}
}
&
\makecell{
\includegraphics[scale=.35]{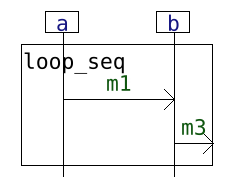}
}\\
\makecell{
\scriptsize\textit{interaction before pruning}
}
&
\makecell{
\scriptsize\textit{pruning with $prune(\_,c)$}
}
&
\makecell{
\scriptsize\textit{interaction after pruning}
}
\end{tabular}
\caption{Illustration of $prune$ (case where a branch of alternative is cut-out)}
\label{fig:pruning_of_alt}
\vspace{-.75cm}
\end{figure}

Fig.\ref{fig:pruning_of_alt} illustrates a specific application of the pruning process. We consider an interaction $i$, drawn on the left part of Fig.\ref{fig:pruning_of_alt} and which term is given in the middle part of Fig.\ref{fig:pruning_of_alt}. $i$ is defined over the set $L=\{a,b,c\}$ of lifelines. We then apply $prune(\_,c)$ on $i$ to obtain the interaction drawn on the right part of Fig.\ref{fig:pruning_of_alt}. 
The blue lines represent the rewriting orchestrated by the $prune$ function.
The only action occurring on $c$ in $i$ is $c?m_2$. It must be eliminated. As its parent is a scheduling operator ($strict$), it must also be eliminated. The grand-parent node is an $alt$ operator. The right cousin underneath this $alt$ is $\varnothing$, which "avoids" $c$. Therefore, we can force the choice of the right branch of this $alt$ to solve the pruning. The remaining interaction then does not contain any action occurring on $c$. As explained earlier, $prune$ made the minimal modifications to $i$ so as to eliminate $c?m_2$. For instance, we could have simply (and naively) forbidden the repetition of the $loop$ at the root position; but this would also have eliminated from the semantics of the remaining interaction a number of traces which we do not want to be eliminated.

\begin{figure}
\vspace{-.75cm}
    \centering
\begin{tabular}{ccc}
\makecell{
\includegraphics[scale=.35]{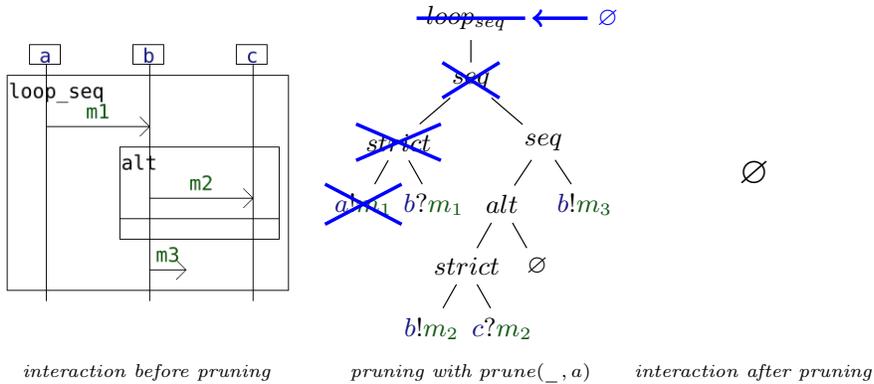}
}
&
\makecell{
\scalebox{1.1}{
\begin{tikzpicture}[every node/.style = {shape=rectangle, align=center}]
\node (o1) { $loop_{seq}$ } [sibling distance=1.75cm,level distance=0.75cm]
child {node (o11) {$seq$} [sibling distance=1.75cm]
      child {node (o111) {$strict$} [sibling distance=.85cm]
        child {node (o1111) {\hlf{$a$}$!$\hms{$m_1$}}}
        child {node (o1112) {\hlf{$b$}$?$\hms{$m_1$}}}
      }
      child {node (o112) {$seq$} [sibling distance=1cm]
        child {node (o1121) {$alt$} [sibling distance=.85cm]
          child {node (o11211) {$strict$} [sibling distance=.85cm]
            child {node (o112111) {\hlf{$b$}$!$\hms{$m_2$}}}
            child {node (o112112) {\hlf{$c$}$?$\hms{$m_2$}}}
          }
          child {node (o11212) {$\varnothing$} }
        }
        child {node (o1122) {\hlf{$b$}$!$\hms{$m_3$}} }
      }
    }
;
\draw[blue,very thick] (o1111.south east) -- (o1111.north west);
\draw[blue,very thick] (o1111.south west) -- (o1111.north east);
\draw[blue,very thick] (o111.south east) -- (o111.north west);
\draw[blue,very thick] (o111.south west) -- (o111.north east);
\draw[blue,very thick] (o11.south east) -- (o11.north west);
\draw[blue,very thick] (o11.south west) -- (o11.north east);
\draw[blue,very thick] (o1.east) -- (o1.west);
\node[blue,thick] (varno) at ([xshift=1cm] o1.east) {$\varnothing$};
\draw[->,blue,very thick] (varno) edge ([xshift=2.5pt] o1.east);
\end{tikzpicture}}
}
&
\makecell{
\Large$\varnothing$
}\\
\makecell{
\scriptsize\textit{interaction before pruning}
}
&
\makecell{
\scriptsize\textit{pruning with $prune(\_,a)$}
}
&
\makecell{
\scriptsize\textit{interaction after pruning}
}
\end{tabular}
\caption{Illustration of $prune$ (case where the repetition of a loop is forbidden)}
\label{fig:pruning_of_loop}
\vspace{-.75cm}
\end{figure}

Fig.\ref{fig:pruning_of_loop} illustrates a specific application of the pruning process. We consider the same interaction $i$ as in Fig.\ref{fig:pruning_of_alt}, However we consider the pruning w.r.t. lifeline $a$ instead of $c$.
The only action occurring on $a$ in $i$ is $a!m_1$. It must be eliminated. Its parent and grand-parent being respectively a $strict$ and a $seq$ operator, they must both be eliminated. Finally, a $loop$ node is reached. At this point, the only choice is to forbid the repetition of this loop. We therefore replace it by the empty interaction $\varnothing$ (as indicated in blue) to obtain the interaction on the right.

We can now define the "e$\chi$ecution" function $\chi$ which computes, from a given interaction $i$ and position $p \in front(i)$, the interaction $i'$ which characterizes all the continuations of the executions of $i$ which start with the execution of action $i_{|p}$ at position $p$.

\begin{definition}[Interaction Execution\label{def:labelled_interaction_execution}]
The function $\chi : Int \times \{1,2\}^* \rightarrow Int \times Act$ , defined for couples $(i,p)$ verifying $p \in front(i)$ is s.t.:
\begin{itemize}
    \item for any $act \in Act$, $\chi(act,\epsilon) = (\varnothing,act)$
    \item for any $i_1,i_2\in Int^2$, $p_1\in front(i_1)$, let us note $\chi(i_1,p_1) = (i'_1,act)$, then:
    \begin{itemize}
        \item $\chi(alt(i_1,i_2),1.p_1)=(i'_1,act)$,
        \item $\chi(f(i_1,i_2),1.p_1)=(f(i'_1,i_2),act)$ for $f \in \{strict,seq,par\}$, 
        \item $\chi(loop_{f}(i_1),1.p_1)=(f(i'_1,loop_{f}(i_1)),act)$ for $f \in \{strict,seq,par\}$,
    \end{itemize}
\item for any $i_1,i_2\in Int^2$, $p_2\in front(i_2)$, let us note $\chi(i_2,p_2) = (i'_2,act)$, then:
    \begin{itemize}
    \item $\chi(alt(i_1,i_2),2.p_2)=(i'_2,act)$
    \item  $\chi(strict(i_1,i_2),2.p_2)=(i'_2,act)$ iff $2.p_2\in front(strict(i_1,i_2))$ 
    \item $\chi(seq(i_1,i_2),2.p_2)=(seq(prune(i_1,lf(act)),i'_2),act)$\\
    iff $2.p_2\in front(seq(i_1,i_2))$ 
    \item $\chi(par(i_1,i_2),2.p_2)=(par(i_1,i'_2),act)$.
    \end{itemize}
\end{itemize}
\end{definition}

$\chi$ is defined by induction on the cases authorized by its precondition $p\in front(i)$. If $i \in Act$, $p$ can only be $\epsilon$ (and vice-versa). In this case $\chi(i,\epsilon)=(\varnothing,i)$ because the action $i$ is executed and nothing remains to be executed. In any other case, $p$ is either of the form $1.p_1$ or $2.p_2$, meaning that the action to be executed is resp. in the left or right sub-interaction. Then the result of $\chi(i,p)$ is a reconstruction of the interaction term from resp. the result of $\chi(i_1,p_1)$ and $i_2$ or the result of $\chi(i_2,p_2)$ and $i_1$.
The most subtle case occurs when $p=2.p_2$ and $i=seq(i_1,i_2)$. The precondition $p\in front(i)$ implies that $i_{|p} \in Act$ and that the left child $i_1$ avoids $lf(i_{|p})$. In this case, to construct $\chi(i,2.p_2)$, $\chi$ does not use $i_1$ but rather its pruned version $prune(i_1,lf(i_{|p}))$ which eliminates all traces involving $lf(i_{|p})$ while preserving all others.

\begin{figure}
%\vspace{-.75cm}
    \centering
\begin{tabular}{ccc}
\makecell{
\includegraphics[scale=.325]{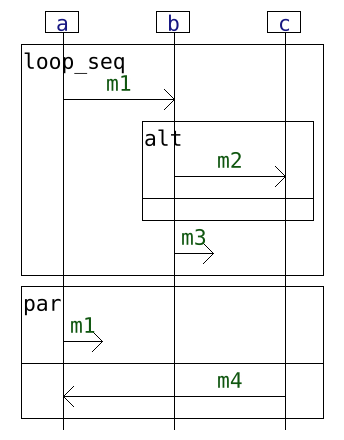}
}
&
\makecell{
\scalebox{.95}{
\begin{tikzpicture}[every node/.style = {shape=rectangle, align=center}]
\node (o) { $seq$ } [sibling distance=1.75cm,level distance=0.75cm]
  child {node (o1) {$loop_{seq}$}
    child {node (o11) {$seq$} [sibling distance=1.75cm]
      child {node (o111) {$strict$} [sibling distance=.85cm]
        child {node (o1111) {\hlf{$a$}$!$\hms{$m_1$}}}
        child {node (o1112) {\hlf{$b$}$?$\hms{$m_1$}}}
      }
      child {node (o112) {$seq$} [sibling distance=1.cm]
        child {node (o1121) {$alt$} [sibling distance=.85cm]
          child {node (o11211) {$strict$} [sibling distance=.85cm]
            child {node (o112111) {\hlf{$b$}$!$\hms{$m_2$}}}
            child {node (o112112) {\hlf{$c$}$?$\hms{$m_2$}}}
          }
          child {node (o11212) {$\varnothing$} }
        }
        child {node (o1122) {\hlf{$b$}$!$\hms{$m_3$}} }
      }
    }
  }
  child {node (o2) {$par$} [sibling distance=.85cm]
    child {node (o21) {\hlf{$a$}$!$\hms{$m_1$}}}
    child {node (o22) {$strict$} [sibling distance=.85cm]
      child {node (o221) {\hlf{$c$}$!$\hms{$m_4$}}}
      child {node (o222) {\hlf{$a$}$?$\hms{$m_4$}}}
    }
  }
;
\begin{pgfonlayer}{background}
\draw[blue,fill=blue,opacity=0.2](o1111.south west) to[closed,curve through={(o1111.west) .. (o1111.north west) .. (o111.south west) .. (o111.west) .. (o111.north west) .. (o11.south west) .. (o11.west) .. (o11.north west) .. (o1.south west) .. (o1.west) .. (o1.north west) .. (o1.north) .. (o1.north east) .. (o1.east) .. (o1.south east) .. (o11.north east) .. (o11.east) ..  (o11.south east) .. (o112.north east) .. (o112.east) .. (o112.south east) .. (o1122.north east) .. (o1122.east) .. (o1122.south east) .. (o1122.south) .. (o11212.north east) .. (o11212.east) .. (o11212.south east) .. (o112112.north east) .. (o112112.east) .. (o112112.south east) .. (o112112.south) .. (o112111.south) .. (o112111.south west) .. (o112111.west) .. (o112111.north west) .. (o1111.south east) .. (o1111.south)}](o1111.south west);
\fill[orange!60] ($(o221) + (-.4,-.25)$) -- ($(o221) + (-.4,+.25)$) -- ($(o221) + (+.4,+.25)$) -- ($(o221) + (+.4,-.25)$) -- cycle;
\end{pgfonlayer};
\draw[blue,very thick] (o112112.south east) -- (o112112.north west);
\draw[blue,very thick] (o112112.south west) -- (o112112.north east);
\draw[blue,very thick] (o11211.south east) -- (o11211.north west);
\draw[blue,very thick] (o11211.south west) -- (o11211.north east);
\draw[blue,very thick] (o1121.east) -- (o1121.west);
\draw[->,blue,very thick] ([xshift=-5pt,yshift=-5pt] o11212.north east) to [bend right=45] ([xshift=2.5pt] o1121.east);
\draw[red,very thick] (o221.south east) -- (o221.north west);
\draw[red,very thick] (o221.south west) -- (o221.north east);
\draw[red,very thick] (o22.east) -- (o22.west);
\draw[->,red,very thick] ([xshift=-5pt,yshift=-5pt] o222.north east) to [bend right=45] ([xshift=2.5pt] o22.east);
\node[draw] (legP) at (1,-4.4) { {\scriptsize \textcolor{blue}{$\blacksquare$} pruning} };
\node[draw,above=.6 of legP.west,anchor=west] (legX) { {\scriptsize \textcolor{red}{$\blacksquare$} e$\chi$ecution} };
\end{tikzpicture}
}
}
&
\makecell{
\includegraphics[scale=.325]{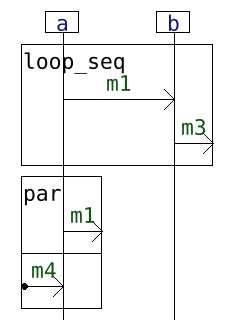}
}\\
\makecell{
\scriptsize\textit{interaction before $\chi$}
}
&
\makecell{
\scriptsize\textit{e$\chi$ecution}
}
&
\makecell{
\scriptsize\textit{interaction after $\chi$}
}
\end{tabular}
\caption{Illustration of e$\chi$ecution}
\label{fig:execution_example}
\vspace{-.75cm}
\end{figure}

Fig.\ref{fig:execution_example} illustrates an application of $\chi$ on the example interaction from Fig.\ref{fig:ex3}. The action $c!m_4$ at the frontier position $221$ is being executed. On the left is the original interaction $i$ and on the right the resulting $i'$ interaction s.t. $\chi(i,p) = (i',c!m_4)$. In the middle is illustrated the process of computing $i'$ from the rewriting of the interaction term $i$.  The first step is the elimination of the target action $c!m_4$ which is done by replacing it with $\varnothing$. Then $i'$ is reconstructed from the bottom to the top. The immediate parent of $c!m_4$ at position $22$ is a $strict$ operator. 
On Fig.\ref{fig:execution_example} we have added term simplification steps which correspond to the elimination of $\varnothing$ children of scheduling operators.
Those simplifications steps are immediate and do not incur changes in the model's semantics as for any scheduling operator $f$ and interaction $i_0$ we have that $f(i_0,\varnothing)$ and $f(\varnothing,i_0)$ are equivalent to $i_0$.
As such, when reaching $i_{|22}$ during the rewriting of $i$ into $i'$, we simply replace the $strict$ node by its right child $a?m_4$. When reaching the $par$ node at position $2$, we do not write any change. However, when reaching the $seq$ node at root position $\epsilon$, the rewriting of $i$ involves the pruning of the left sub-interaction $i_{|1}$ (highlighted in blue) w.r.t. the lifeline on which the executed action $c!m_4$ occurs (i.e. $c$).

\subsection{Definition of accepted (multi-)traces \label{subsec:multitr_semantics}}

Our small-step approach then consists in the exploration of an execution tree representing all possible successions of transformations \resizebox{!}{10pt}{$i\xrightarrow{act@p}i'$}, starting from an initial interaction $i_0$. An accepted trace corresponds to a sequence $act_1.\cdots.act_n$ obtained from 
a path \resizebox{!}{10pt}{$i_0\xrightarrow{act_1@p_1}i_1 \cdots \xrightarrow{act_n@p_n}i_n$} with 
$i_n$ a terminal interaction, i.e. accepting $\epsilon$. By grouping all such paths together, we obtain a tree, called the execution tree, whose nodes are interactions and arcs are labelled by couples $(p,act)$ noted $act@p$. For a node $i$, child nodes are interactions $i'$ obtained via the execution of any frontier action $act=i_{|p}$ with $p\in front(i)$. Any such child node $i'$ corresponds to an interaction that accepts traces that are suffixes of traces accepted by $i$ and which start with $act$. Let us note that, given the existence of loops, execution trees can be infinite, and traces can be arbitrarily long.

\begin{figure}[!ht]
%\vspace*{-.75cm}
\centering
\includegraphics[width=.6\textwidth]{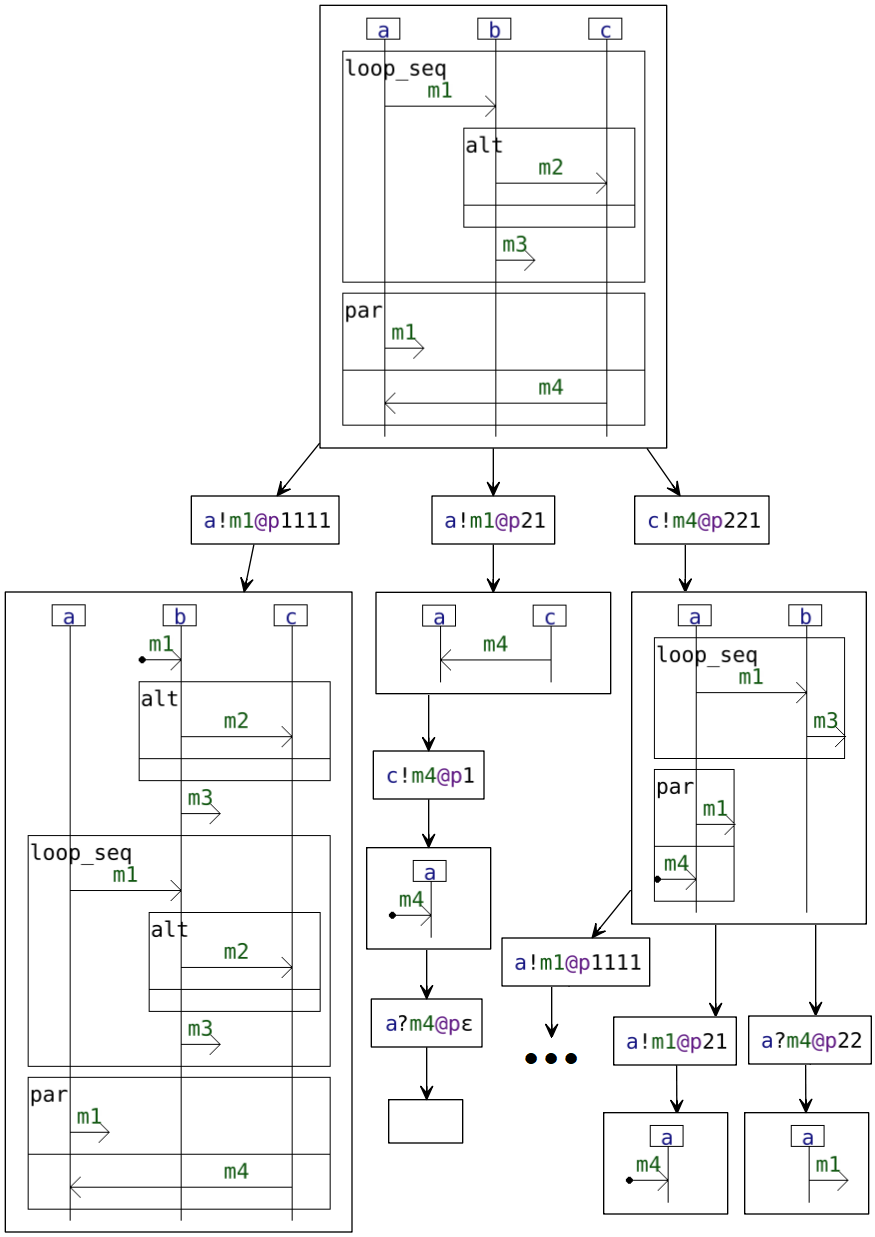}
\caption{Partial exploration of an interaction's trace semantics using $\chi$ as small-step}
\label{fig:partial_explo}
\vspace*{-.75cm}
\end{figure}

On Fig.\ref{fig:partial_explo} is illustrated such an execution tree. It is only partially drawn given that, in any case, the original interaction having a $loop$ node, this tree is infinite. Let us note that the transformation \resizebox{!}{10pt}{$i \xrightarrow{c!m_4@221} \_$} in which the interaction on the top part becomes the one underneath on its right (after transition "\hlf{$c$}$!$\hms{$m_4$}\hfresh{@}$221$"), corresponds to the e$\chi$ecution illustrated on Fig.\ref{fig:execution_example}. The 3 child interaction underneath the top one correspond to the e$\chi$ecutions of the $3$ frontier actions of $i$ (detailed in Fig.\ref{fig:frontier_ex}). The path leading to the empty interaction (white square $\square$) yields the trace $a!m_1.c!m_4.a?m_4$ (which is therefore part of the trace semantics of $i$). 

In Def.\ref{def:operational_semantics} below, we formally define our trace and multi-trace semantics $Accept$ and $AccMult$.

\begin{definition}[Semantics\label{def:operational_semantics}]
The sets of accepted traces $Accept : Int \rightarrow \mathcal{P}(Act^*)$ and multi-traces $AccMult : Int \rightarrow \mathcal{P}(Mult)$
are s.t. for any $i \in Int$:
\[
Accept(i) = empty(i) \cup
\left\{ ~
act.Accept(i') ~|~
\exists ~ p \in front(i), ~ \chi(i,p) = (i',act) ~
\right\}
\]
\[AccMult(i) = \{proj(\varsigma) ~|~ \varsigma \in Accept(i) \}\]

with: $empty(i) = \{\epsilon\}$ if $exp_\epsilon(i)$ and $empty(i) = \emptyset$ otherwise.
\end{definition}

\section{Multi-trace membership analysis\label{sec:analysis}}

\subsection{Principle\label{sec:ana_principle}}

We define a process able to decide whether or not a multi-trace $\mu$ is accepted by an interaction $i$. Its key principle is to construct traces accepted by $i$ that project on $\mu$. Constructing those traces is based on elementary steps $(i,\mu) \leadsto (i',\mu')$ s.t. $\chi(i,p) = (i',act_j)$ for some $p\in front(i)$ with $act_j \in Act(l_j)$
and $\mu$ and $\mu'$ being of resp. forms $(\sigma_1,\cdots,act_j.\sigma_j,\cdots,\sigma_n)$ and  
$(\sigma_1,\cdots,\sigma_j,\cdots,\sigma_n)$.

Any such step $(i,\mu) \leadsto (i',\mu')$ corresponds simultaneously, for a given action $act$, (1) to a small-step e$\chi$ecution in $i$ of $act$ (there can be several of those given that a same action can be present at multiple positions) and (2) to the consumption of $act$ on the component $lf(act)$ of $\mu$ (reducing its size by $1$).

By considering all possible matches between frontier actions $i_{|p}$ and the heads of components of $\mu$, and by iterating those steps of computation, the process builds a tree whose paths are of the form $(i_0,\mu_0) \leadsto \cdots \leadsto (i_p,\mu_p) \cdots \leadsto (i_q,\mu_q)$, denoted as  $(i_0,\mu_0)$ \resizebox{!}{10pt}{$\overset{*}{\leadsto}$}  $(i_q,\mu_q)$.

At each step $(i,\mu) \leadsto (i',\mu')$, the size of the multi-trace decreases by one. Hence, any path eventually reaches a point where it is no longer possible to find a next step. This halting of the process can occur in 2 cases.

(1) Either the process reaches a state $(i_q,\mu_q)$ where $\mu_q$ is not empty and no frontier action of $i_q$ matches some first elements in $\mu_q$. In that case the sequence of actions that leads to $(i_q,\mu_q)$ is not a trace accepted by $i$ and a local verdict $UnCov$ (for "multi-trace not covered") is associated to $(i_0,\mu_0)$ \resizebox{!}{10pt}{$\overset{*}{\leadsto}$}  $(i_q,\mu_q)$.

(2) Or the process reaches a state $(i_q,(\epsilon,\cdots ,\epsilon))$. Here, all actions of $\mu$ have been consumed to form a given global trace $\varsigma$. The process then checks if $\varsigma$ is accepted by $i$ (which happens iff $exp_\epsilon(i_q)$). If the answer is yes then $(i_0,\mu_0)$ \resizebox{!}{10pt}{$\overset{*}{\leadsto}$}  $(i_q,(\varepsilon,\cdots ,\varepsilon))$ is associated with a coverage verdict $Cov$ (for "multi-trace covered"). Otherwise, the verdict $UnCov$ is associated to the path. 

If there exists a path leading to $Cov$, the global verdict is $Pass$. If no such path exists, the global verdict is $Fail$.

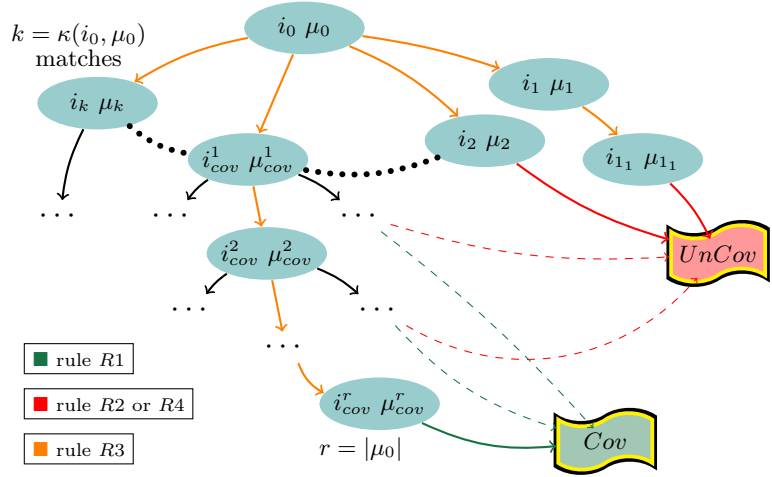
\begin{figure}
\vspace*{-.75cm}
    \centering
\begin{tikzpicture}[
local_verdict/.style = {
              draw=yellow,
              minimum width=1.25cm,
              outer ysep=-.5mm,
              inner sep=2mm,
              line width=0.5mm,
              append after command={\pgfextra{
                    \node [#1,
                           draw=black, 
                           line width=0.75mm,
                           inner sep=0mm,
                           fit=(\tikzlastnode)] {};
              }}}
]
% ==============================================
\tikzstyle{uncov}=[tape,local_verdict=tape,fill=red!40,label=center:$UnCov$]
\tikzstyle{cov}=[tape,local_verdict=tape,fill=darkspringgreen!40,label=center:$Cov$]
\tikzstyle{vertex}=[ellipse,fill=teal!40,minimum width=1.6cm,minimum height=.7cm]
% ==============================================
\node [vertex,label=center:$i_0~\mu_0$] (x0) at (0,0)   {};
\node [vertex,label=center:$i_1~\mu_1$] (x1) at (3.25,-.75)  {};
\node [vertex,label=center:$i_{1_1}~\mu_{1_1}$] (x11) at (4.5,-1.75)  {};
\node [vertex,label=center:$i_2~\mu_2$] (x2) at (2.4,-1.5)    {};
%\node [vertex,label=center:$i_{2\cdot 1}~\mu_{2\cdot 1}$] (x21) at (3.25,-2.5)  {};
\node[uncov] (u) at (5.5,-3) {};
\node[cov] (c) at (4,-5.5) {};
\node [vertex,label=center:$i_k~\mu_k$] (xk) at (-2.75,-1)  {};
\node (xkleg) at  (-3,-.1) {{\footnotesize $k=\kappa(i_0,\mu_0)$}};
\node (xkleg2) at (-3,-.4) {{\footnotesize matches}};
\node (anchor_xk) at (-3.25,-2.5) {{\large $\cdots$}};
\node [vertex,label=center:$i_{cov}^1~\mu_{cov}^1$] (xw1) at (-.75,-1.75)  {};
\node (left_xw1) at (-1.75,-2.5) {{\large $\cdots$}};
\node (right_xw1) at (.75,-2.5) {{\large $\cdots$}};
\node [vertex,label=center:$i_{cov}^2~\mu_{cov}^2$] (xw2) at (-.5,-3)  {};
\node (xwdots) at (-.25,-4.25) {{\large $\cdots$}};
\node (left_xw2) at (-1.5,-3.75) {{\large $\cdots$}};
\node (right_xw2) at (1,-3.75) {{\large $\cdots$}};
\node [vertex,label=center:$i_{cov}^r~\mu_{cov}^r$] (xwr) at (1,-5)  {};
\node (xwrleg) at (.75,-5.6) {{\footnotesize $r=|\mu_0|$}};
%\node (left_dots) at (-3.25,-2.5) {{\large $\cdots$}};
% ==============================================
\draw[line width=2.25pt, line cap=round, dash pattern=on 0pt off 2\pgflinewidth] (xk) edge[bend right=15] (xw1);
\draw[line width=2.25pt, line cap=round, dash pattern=on 0pt off 2\pgflinewidth] (x2) edge[bend left=15] (xw1);
% ==============================================
\draw[thick,orange,->] (x0) edge[bend left=5] (x1);
\draw[thick,orange,->] (x1) edge[bend left=5] (x11);
\draw[thick,orange,->] (x0) edge[bend left=10] (x2);
%\draw[thick,orange,->] (x2) edge[bend left=10] (x21);
\draw[thick,orange,->] (x0) edge[bend right=10] (xk);
\draw[thick,->] (xk) edge[bend right=10] (anchor_xk);
\draw[thick,orange,->] (x0) edge (xw1);
\draw[thick,orange,->] (xw1) edge (xw2);
\draw[thick,->] (xw1) edge[bend right=10] (left_xw1);
\draw[thick,->] (xw1) edge[bend left=10] (right_xw1);
\draw[thick,orange,->] (xw2) edge (xwdots);
\draw[thick,->] (xw2) edge[bend right=10] (left_xw2);
\draw[thick,->] (xw2) edge[bend left=10] (right_xw2);
\draw[thick,orange,->] (xwdots) edge[bend right=20] (xwr);
% ==============================================
\draw[thick,red,->] (x11) edge[bend left=10] (u);
%\draw[thick,red,->] (x21) edge[bend right=10] (u);
\draw[thick,red,->] (x2) edge[bend right=10] (u);
\draw[thick,darkspringgreen,->] (xwr) edge[bend right=15] (c);
% ==============================================
\draw[dashed,red,->] (right_xw2) edge[bend right=40] (u);
\draw[dashed,darkspringgreen,->] (right_xw2) edge[bend right=15] (c);
% ==============================================
\draw[dashed,red,->] (right_xw1) edge[bend right=10] (u);
\draw[dashed,darkspringgreen,->] (right_xw1) edge[bend left=5] (c);
% ==============================================
\node[draw] (legR1) at (-3,-4.4) { {\scriptsize \textcolor{darkspringgreen}{$\blacksquare$} rule $R1$} };
\node[draw,below=.6 of legR1.west,anchor=west] (legR2or4) { {\scriptsize \textcolor{red}{$\blacksquare$} rule $R2$ or $R4$} };
\node[draw,below=.6 of legR2or4.west,anchor=west] (legR3) { {\scriptsize \textcolor{orange}{$\blacksquare$} rule $R3$} };
\end{tikzpicture}
    \caption{Principle of multi-trace analysis}
    \label{fig:principle_ana}
\vspace*{-.75cm}
\end{figure}

Let us consider the illustration on Fig.\ref{fig:principle_ana}. Starting from node $(i_0,\mu_0)$, a number of paths can be explored.
From $(i_0,\mu_0)$, there exists $k$ outgoing transitions to other nodes, $k$ being the number of matches between frontier actions of $i_0$ and trace actions at the heads of components from $\mu$. Exploration steps (which are not represented but implicitly designated by $\cdots$ in Fig.\ref{fig:principle_ana}) are then repeated (recursively) for every one of those children.
Ultimately, every path that is thus created leads back to one of the two coverage verdicts $UnCov$ or $Cov$ (given the decreasing size of the multi-trace).

Paths starting from $(i_0,\mu_0)$ may have different lengths and different outcomes. This is explained by the fact that the graph explores how some executions of $i_0$ might (or might not) cover the behavior expressed by the multi-trace $\mu_0$. It may be so that there exists several executions of $i_0$ that match $\mu_0$. At the same time there might exists some that do not, and the fact that they do not match $\mu_0$ can be made clear after an arbitrary number (bounded by the length of $\mu_0$) of small-step e$\chi$ecutions.

With regard to $Cov$, the fact that several paths might lead to it may be explained by the fact that several global traces can be projected into the same multi-trace (as in Fig.\ref{fig:semantics_enumeration_ex}). Therefore, when trying to reorder $\mu_0$ into a global trace that satisfies $i_0$, we can find several of those.

In the example illustrated in Fig.\ref{fig:principle_ana}, there exists (at least one) such path $[(i_0,\mu_0) \leadsto (i_{cov}^1,\mu_{cov}^1) \leadsto \cdots \leadsto (i_{cov}^r,\mu_{cov}^r)]$ that leads to $Cov$. Given that obtaining $Cov$ requires to empty the initial multitrace $\mu_0$, the length of this path $r$ is equal to that of the multi-trace. The existence of this path then implies that the global verdict $Pass$ will be returned.

%As there are only two coverage verdicts, different paths lead to the same one. 
%(e.g. some interleavings may be transparent to multitraces). Let us also remark that all paths leading to $Cov$ are of the same length (that of the initial multitrace: $|\mu_0|$). However paths leading to $UnCov$ are of any length smaller or equal to $|\mu_0|$. 

\subsection{Definition of analysis process\label{sec:ana_def}}

Multi-trace analysis relies on 4 rules, denoted $R1$, $R2$, $R3$ and $R4$ and given in Def.\ref{def:ana_rules}.
Those rules define a directed graph $\mathbb{G}$ in which vertices are either a tuple $(i,\mu) \in Int \times Mult$ or a coverage verdict $v \in \{Cov,UnCov\}$. We note $\mathbb{V} = \{Cov,UnCov\} \cup (Int \times Mult)$ the set of vertices.
For $x$ in $\{1,2,3,4\}$, the rule $(Rx) \frac{v}{v'} \; cond$, with $v \in Int \times Mult$ and $v' \in \mathbb{V}$  specifies edges of the form $v \leadsto v'$ of that graph, provided that $v$ satisfies condition $cond$.

\begin{definition}[Rules of Multi-Trace Analysis\label{def:ana_rules}]
\\The analysis relation 
$\leadsto \subseteq \mathbb{V} \times \mathbb{V}$ is defined as:

\vspace*{-.25cm}

\begin{minipage}[t]{.475\textwidth}
\begin{prooftree}
\AxiomC{$i$}
\AxiomC{$(\epsilon,\cdots,\epsilon)$}
\LeftLabel{(R1)}
\RightLabel{$exp_\epsilon(i)$}
\BinaryInfC{$Cov$}
\end{prooftree}
\end{minipage}
\begin{minipage}[t]{.475\textwidth}
\begin{prooftree}
\AxiomC{$i$}
\AxiomC{$(\epsilon,\cdots,\epsilon)$}
\LeftLabel{(R2)}
\RightLabel{$\neg exp_\epsilon(i)$}
\BinaryInfC{$UnCov$}
\end{prooftree}
\end{minipage}

\begin{minipage}[t]{.9\textwidth}
\begin{prooftree}
\AxiomC{$i$}
\AxiomC{$(\sigma_1,\cdots,act.\sigma_k,\cdots,\sigma_n)$}
\LeftLabel{(R3)}
\RightLabel{$\exists ~ p \in front(i) \text{ s.t. } \chi(i,p) = (i',act)$}
\BinaryInfC{$i' \hspace{0.5cm} (\sigma_1,\cdots,\sigma_k,\cdots,\sigma_n)$}
\end{prooftree}
\end{minipage}

\vspace*{.25cm}

\begin{minipage}[t]{.9\textwidth}
\begin{prooftree}
\AxiomC{$i$}
\AxiomC{$(\sigma_1,\cdots,\sigma_n)$}
\LeftLabel{(R4)}
\RightLabel{$
\left\{
\begin{array}{l}
(\sigma_1,\cdots,\sigma_n) \neq (\epsilon,\cdots,\epsilon)\\
\wedge 
\left(
\begin{array}{l}
\forall~ j \in [1,n],~\forall~ p \in front(i),\\
(\sigma_j \neq \epsilon) \Rightarrow (fst(\sigma_j) \neq i_{|p})
\end{array}
\right)
\end{array}
\right.
$}
\BinaryInfC{$UnCov$}
\end{prooftree}

\vspace*{.2cm}

where $fst(\sigma)$ denotes the first element of a non empty sequence $\sigma$.
\end{minipage}

\end{definition}

Vertices of the form $(i,\mu)$ are not sinks.
Indeed, if $\mu$ is the empty multi-trace, given that $exp_\epsilon(i)$ can either be $true$ or $false$, either $R1$ or $R2$ applies and therefore there exists an outgoing edge from any $(i,(\epsilon,\ldots,\epsilon))$. 
If $\mu \neq (\epsilon,\ldots,\epsilon)$, one can either have or not have matches between frontier actions and multi-trace component heads. Hence, an outgoing edge exists according to either $R3$ or $R4$. Consequently, coverage verdicts $\{Cov,UnCov\}$ are the $2$ only sinks of graph $\mathbb{G}$.

Rules $R1$, $R2$ and $R4$ specify edges from vertices of the form $(i,\mu)$ to coverage verdicts. The rule $R3$ specifies edges $(i,\mu) \leadsto (i',\mu')$ such that 
(1) there exists an action $act$ occurring in $i$ at position $p \in front(i)$ matching a head action $act_j$ of $\mu$, i.e. $\mu=(\sigma_1,\cdots,act_j.\sigma_j',\cdots,\sigma_n)$, 
(2) $i'$ is defined by $\chi(i,p) = (i',act_j)$, 
and (3) $\mu'$ is the multi-trace $\mu$ in which we have removed $act_j$, i.e. $\mu' =(\sigma_1,\cdots,\sigma_j',\cdots,\sigma_n)$.
Moreover, let us note that for a vertex $(i,\mu)$, there are at most $|front(i)|$ possible applications of the rule $R3$ with $|front(i)|$ bounded by the number of occurrences of actions in $i$.

Let us consider $|\mu|$ the number of actions occurring in a multi-trace $\mu$, i.e. the sum of lengths of its component traces. Let us extend this notation to vertices, that is, $|(i,\mu)|$ defined as $|\mu|$, and $|Cov|$ and $|UnCov|$ defined as $-1$.
For any edge $v \leadsto v'$ of $\mathbb{G}$,  we have $|v'| < |v|$ with $|v'| \geq -1$. Consequently, the successive application of the rules strictly decrements the size of nodes and from any vertex $(i,\mu)$, any maximal outgoing path is finite, and terminates in a coverage verdict in $\{Cov,UnCov\}$ (since $(i,\mu)$ are not sinks of $\mathbb{G}$). 
Thus, $\mathbb{G}$ is an acyclic graph. 
With the notation $v$ \resizebox{!}{10pt}{$\overset{*}{\leadsto}$} $v'$ to indicate that there is a path from $v$ to $v'$ in $\mathbb{G}$,
we define multi-trace analysis.

\begin{definition}[Multi-Trace Analysis\label{def:ana_algo}]
We define $\omega : Int \times Mult \rightarrow \{Pass,Fail\}$ such that for any $i \in Int$ and $\mu \in Mult$ we have:
\begin{itemize}
    \item $\omega(i,\mu) = Pass$ iff there exists a path $(i,\mu)$ \resizebox{!}{10pt}{$\overset{*}{\leadsto}$} $Cov$
    \item $\omega(i,\mu) = Fail$ otherwise; \\
    i.e. for all path $(i,\mu)$ \resizebox{!}{10pt}{$\overset{*}{\leadsto}$} $v$ with $v \in \{Cov,UnCov\}$, then $v=UnCov$
\end{itemize}
\end{definition}

The function $\omega$ is well-defined. Indeed, we established that any maximal path from a vertex $(i_0,\mu_0)$ has a maximum length of $|\mu|+1 $ and end on a coverage verdict ($Cov$ or $UnCov$). Besides, each intermediate vertex $(i,\mu)$ between $(i_0,\mu_0)$ and a coverage verdict has a number of children bounded by the number of actions of $i$. Therefore, the number of vertices reachable from $(i_0,\mu_0)$ is finite 

All definitions of Sec.\ref{sec:semantics} are structured by induction on terms and positions, and as such, allow a direct implementation of the execution function $\chi$ (involved in the application of rule $R3$). Similarly, one can implement the $\omega$ function by building on-the-fly the sub-graph originating from $(i,\mu)$ thanks to queues and usual search heuristics.
Moreover, in practice, graph traversals can be interrupted as soon as a $Cov$ verdict is reached. 
We have implemented $\chi$ and $\omega$ in the HIBOU tool available online\footnote{\url{https://github.com/erwanM974/hibou_label}}.
In HIBOU, traceability for end-users is facilitated by the operational nature of our approach.
Indeed, whether we are exploring an execution tree (as in Fig.\ref{fig:partial_explo}) or analyzing a multi-trace (as in Fig.\ref{fig:analysis_example}), HIBOU has access to the succession of states and can draw representations of the process.

In Fig.\ref{fig:analysis_example} below is represented the analysis of multitrace $\mu = (a!m_1.a?m_4,~\epsilon,~c!m_4)$ w.r.t. the interaction from Fig.\ref{fig:ex3}, which yields the global verdict $Pass$. Here we configured HIBOU to use a Depth First Search heuristic when exploring the graph $\mathbb{G}$. Given that $Cov$ was found quickly, paths starting with the executions of frontier action $c!m_4$ have not been explored.

\begin{figure}
    \centering
    \includegraphics[width=.7\textwidth]{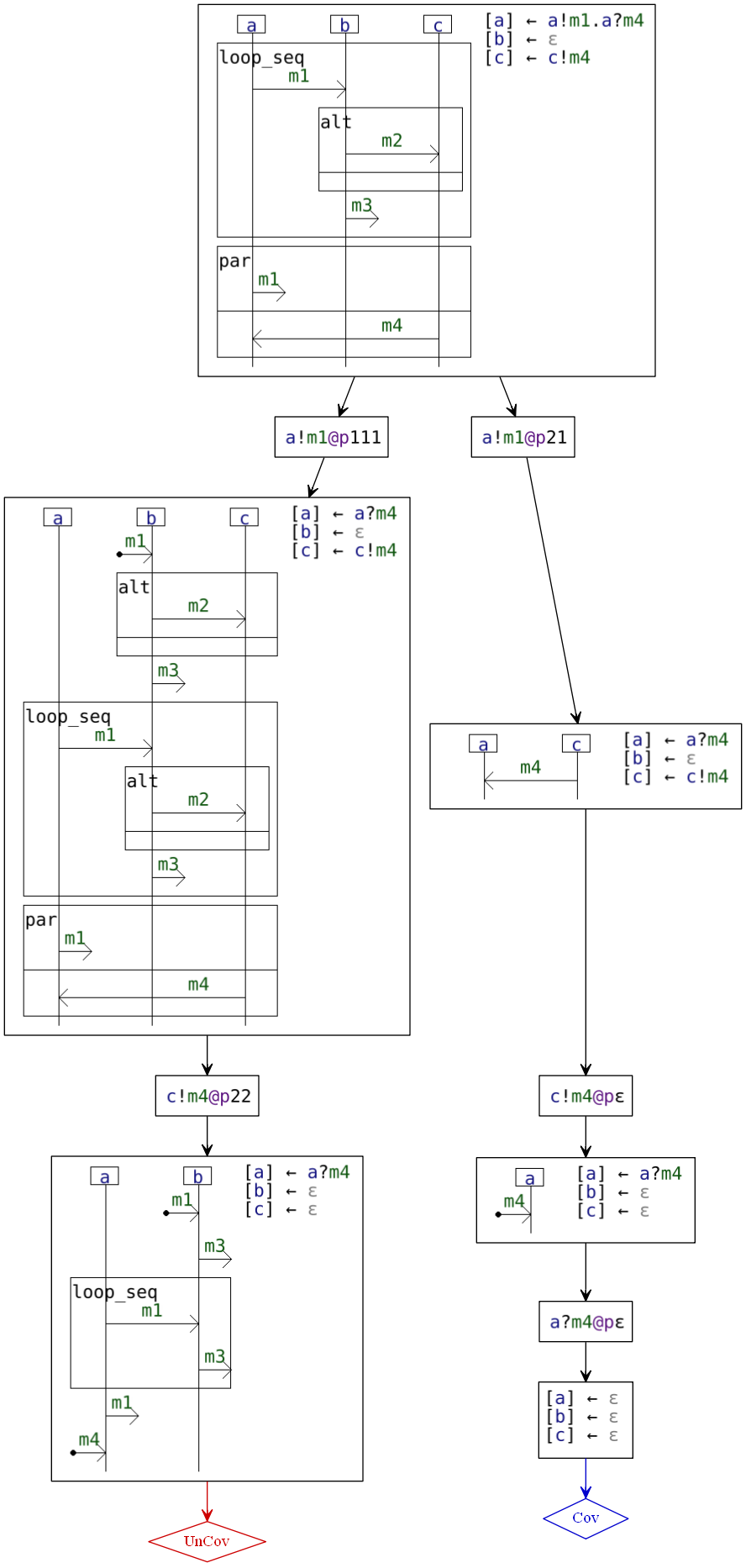}
    \caption{Example of multi-trace analysis}
    \label{fig:analysis_example}
\end{figure}

\clearpage

\subsection{Correctness w.r.t the semantics\label{sec:ana_charac}}

We now prove that the function $\omega$ in charge of analysing multi-traces w.r.t. an interaction
captures exactly its semantics defined by the step-by-step execution function $\chi$ given in Sec.\ref{sec:semantics}. More precisely, we will prove that for any $(i,\mu)$ in $Int \times Mult$, $\omega(i,\mu) = Pass$ iff $\mu \in AccMult(i)$ (and by extension,  $\omega(i,\mu) =  Fail$ iff $\mu \not\in AccMult(i)$). Given that $AccMult(i)$ is the set of projected global traces of $Accept(i)$, it then suffices to prove that for any trace $\varsigma \in Act^*$ we have $\omega(i,proj(\varsigma))=Pass$ iff $\varsigma \in Accept(i)$. Below, Th.\ref{th:accept_implies_pass} and Th.\ref{th:pass_implies_accept} resp. correspond to the $\Leftarrow$  and  $\Rightarrow$ implication of this "\emph{iff}".

\begin{theorem}[$Accept$ implies $Pass$\label{th:accept_implies_pass}]
For any $i \in Int$ and $\varsigma \in Act^*$:
\[
(\varsigma\in Accept(i)) \Rightarrow (\omega(i,proj(\varsigma)) = Pass)
\]
\end{theorem}

\begin{proof}
Let us reason by induction on the trace $\varsigma$.\\
\noindent $\bullet$  $\varsigma=\epsilon$. Let us consider an interaction $i$ s.t. $\epsilon \in Accept(i)$. We have $proj(\epsilon) = (\epsilon,\cdots,\epsilon)$.  As $\epsilon \in Accept(i)$, then $exp_\epsilon(i) = \top$ and $R_1$ is applicable from $(i,(\epsilon,\cdots,\epsilon))$. We obtain $\omega(i,(\epsilon,\cdots,\epsilon)) = Pass$.

\noindent $\bullet$ $\varsigma = act.\varsigma'$. Let us consider $i$ s.t. $\varsigma \in Accept(i)$. The induction hypothesis on $\varsigma'$ is:  "$\forall~ i' \in Int,~ (\varsigma' \in Accept(i')) \Rightarrow (\omega(i',proj(\varsigma')) = Pass)$". As $act.\varsigma' \in Accept(i)$, then there exists $i'$ in $Int$ and $p \in front(i)$ s.t. $\chi(i,p)=(i',act)$ and $\varsigma' \in Accept(i')$.   Let us consider the index $j$ such that $proj(act.\varsigma') = (\sigma_1,\cdots,act.\sigma_j,\cdots,\sigma_n)$. Given that $\chi(i,p)=(i',act)$, $R3$ can be applied so that  $(i,(\sigma_1,\cdots,act.\sigma_j,\cdots,\sigma_n)) \leadsto (i',(\sigma_1,\cdots,\sigma_j,\cdots,\sigma_n))$ with $(\sigma_1,\cdots,\sigma_j,\cdots,\sigma_n) = proj(\varsigma')$. By induction, we have $(\omega(i',proj(\varsigma')) = Pass)$, i.e. there exists a path $(i',proj(\varsigma'))$ \resizebox{!}{10pt}{$\overset{*}{\leadsto}$} $Cov$. By preceding this path with $(i,proj(act.\varsigma')) \leadsto (i',proj(\varsigma'))$, we get a path $(i,(\sigma_1,\cdots,act.\sigma_j,\cdots,\sigma_n))$ \resizebox{!}{10pt}{$\overset{*}{\leadsto}$} $Cov$ and $\omega(i,proj(\varsigma)) = Pass$.
\qed
\end{proof}

\begin{theorem}[$Pass$ implies $Accept$\label{th:pass_implies_accept}]
For any $i \in Int$ and $\mu \in Mult$:
\[
(\omega(i,\mu) = Pass)
\Rightarrow 
\left( 
\exists~ \varsigma \in Act^* \text{ s.t. } proj(\varsigma) = \mu \text{ and } \varsigma \in Accept(i)
\right)
\]
\end{theorem}

\begin{proof}
Let us reason by induction on the size of $\mu$, i.e. on $|\mu|$.\\
\noindent $\bullet$ $|\mu| = 0$. Let us consider $i$ s.t. $\omega(i,\mu) = Pass$. By $|\mu|=0$, $\mu = (\epsilon,\cdots,\epsilon)$. 
Since $\omega(i,(\epsilon,\cdots,\epsilon)) = Pass$, $R1$ must apply and this implies that $exp_\epsilon(i) = \top$ and consequently $\epsilon \in Accept(i)$. Therefore the property holds at length 0.\\
\noindent $\bullet$ $|\mu| = z+1$ with $z \geq 0$. Let us consider $i$ s.t. $\omega(i,\mu) = Pass$. The induction hypothesis states that "for all $(i',\mu') \in Int\times Mult$ with $|\mu'| = z$, $(\omega(i',\mu') = Pass)$ $\Rightarrow$
$( 
\exists~ \varsigma' \in Act^* \text{ s.t. } proj(\varsigma')= \mu' \text{ and } \varsigma' \in Accept(i')
)$".
Since $\omega(i,\mu) = Pass$, there exists a path $(i,\mu)$
\resizebox{!}{10pt}{$\overset{*}{\leadsto}$}
$Cov$. As noticed in Sec.~\ref{sec:ana_def}, each edge of a maximal path exactly consumes one action, with the exception of the last edge leading to the coverage verdict. Thus the path starts with an edge of form $(i,\mu) \leadsto (i',\mu')$ with $|\mu'| = z$ and we have then $(i',\mu')$ \resizebox{!}{10pt}{$\overset{*}{\leadsto}$} $Cov$.  By definition, $\omega(i',\mu')=Pass$. By induction, there exists a trace $\varsigma'$ s.t. $proj(\varsigma')= \mu'$ and $\varsigma' \in Accept(i')$.
$(i,\mu) \leadsto (i',\mu')$  corresponds to the consumption of an action $act$ which matches a frontier action $i_{|p}$ of $i$.
By definition, the trace $\varsigma = act.\varsigma'$ verifies $proj(\varsigma) = \mu$ and $\varsigma \in Accept(i)$.
\qed
\end{proof}

The two theorems demonstrate that $\omega(i,\mu) = Pass$ characterizes the membership of a multi-trace $\mu$ to $AccMult(i)$. Those theorems and all the definitions and lemmas they depend on have been encoded in the Gallina language so as to formally verify our proofs using the Coq automated theorem prover. A Coq proof, which includes the 2 previous demonstrations is available online\footnote{\url{https://erwanm974.github.io/coq_hibou_label_multi_trace_analysis/}}.

The computational cost of $\omega$ varies greatly depending on the initial $(i,\mu)$ couple. In the following we demonstrate the NP-hardness of this membership problem through a reduction of the 1-in-3-SAT problem \cite{1in3SAT78}. This discussion is inspired by \cite{AlurEY01,GenestM08,Hierons14}.

\subsection{Discussion on Complexity}

1-in-3-SAT~\cite{1in3SAT78} is a particular Boolean satisifiability problem. Let us consider a set of $p\geq 1$ boolean variables $V=\{v_1,\cdots,v_p\}$ and a set of $q\geq 1$ clauses $\{C_1,\cdots,C_q\}$ in 3-CNF form i.e. s.t. for any $j\in [1,q]$, $C_j= \alpha_j \vee \beta_j \vee \gamma_j $ with $\alpha_j,\beta_j,\gamma_j$ in $V\cup \overline{V}$, $\bar{~}$ being the usual negation operator\footnote{canonically extended to any set of formulas $X$ as $\overline{X}=\{\overline{\psi}| \psi \in X\}$}. The 1-in-3-SAT problem on formula $\phi = C_1 \wedge \cdots \wedge C_q$ then consists in finding $\rho : V \rightarrow \{\top,\bot\}$ s.t.\footnote{"$\rho \models \phi$" is the  usual satisfaction relation in propositional logic.} $\rho \models \phi$ and s.t. for any clause $C_j$, only one in the three literals $\alpha_j$, $\beta_j$, or $\gamma_j$ is set to $\top$. In the following, we sketch a reduction proof which states that any 1-in-3-SAT problem can be reduced to the multi-trace membership problem for a given $(i,\mu) \in Int\times Mult$ (i.e. whether or not $\mu \in AccMult(i)$).

Let us consider the reduction of 1-in-3-SAT in the simple case where $p=4$ and $q=2$.
This approach can then be extended to include any other case.

From formula $\phi=C_1 \wedge C_2$, we define an interaction $i$ via a 1-on-1 transformation. This $i$ is of the form exemplified on Fig.\ref{fig:ex_1in3sat_simple} i.e. a parallelisation of 4 alternatives $alt(i_v,i_{\overline{v}})$
s.t. for any $x \in V\cup \overline{V}$, $i_x$ is s.t. if $x$ occurs: \\
\noindent $\bullet$ in $C_1$ and $C_2$ then $i_x=seq(l_1!m,l_2!m)$\\
\noindent $\bullet$ in $C_1$ but not in $C_2$ then $i_x=l_1!m$\\
\noindent $\bullet$ in $C_2$ but not in $C_1$ then $i_x=l_2!m$\\
\noindent $\bullet$ neither in $C_1$ nor in $C_2$ then $i_x=\varnothing$

For instance, with $C_1=(v_1 \vee \overline{v_2} \vee v_4)$ and $C_2=(v_1\vee v_3 \vee \overline{v_4})$,  Fig.\ref{fig:ex_1in3sat_simple} gives the corresponding interaction.

We affirm that this 1-in-3-SAT problem $\phi$ is equivalent to the multi-trace membership problem $\mu=(l_1!m, l_2!m) \in AccMult(i)$. Indeed, in a given execution of $i$, component $\sigma_1=l_1!m$ of $\mu$ is expressed exactly once iff exactly one of the sub-interactions $i_{\alpha_1}$, $i_{\beta_1}$ or $i_{\gamma_1}$ is "chosen" during the execution of $i$. Given that the parent interaction (within $i$) of sub-interaction $i_{\alpha_1}$ (same reasoning for $i_{\beta_1}$ and $i_{\gamma_1}$) is of the form $alt(i_{\alpha_1},i_{\overline{\alpha_1}})$ (or with the order of branches inverted), what we mean by "chosen" is that the exclusive branch that hosts $i_{\alpha_1}$ is chosen over that which hosts $i_{\overline{\alpha_1}}$.

The expression of component $\sigma_1$ on lifeline $l_1$ is therefore equivalent to the satisfaction of clause $C_1$ in 1-in-3-SAT. In our example, with $C_1=(v_1 \vee \overline{v_2} \vee v_4)$, the fact that $\rho \models C_1$ with $\rho:[v_1 \rightarrow \bot, v_2 \rightarrow \top, v_3 \rightarrow \top, v_4 \rightarrow \top]$ is equivalent to the fact that $l_1!m$ is expressed exactly once during the execution of $i$ when $i_{\overline{v_1}}$ is chosen over $i_{v_1}$,  $i_{v_2}$ over $i_{\overline{v_2}}$, $i_{v_3}$ over $i_{\overline{v_3}}$, and $i_{v_4}$ over $i_{\overline{v_4}}$.

\begin{wrapfigure}{l}{0.28\textwidth}
\vspace{-.75cm}
\centering
\includegraphics[width=.3\textwidth]{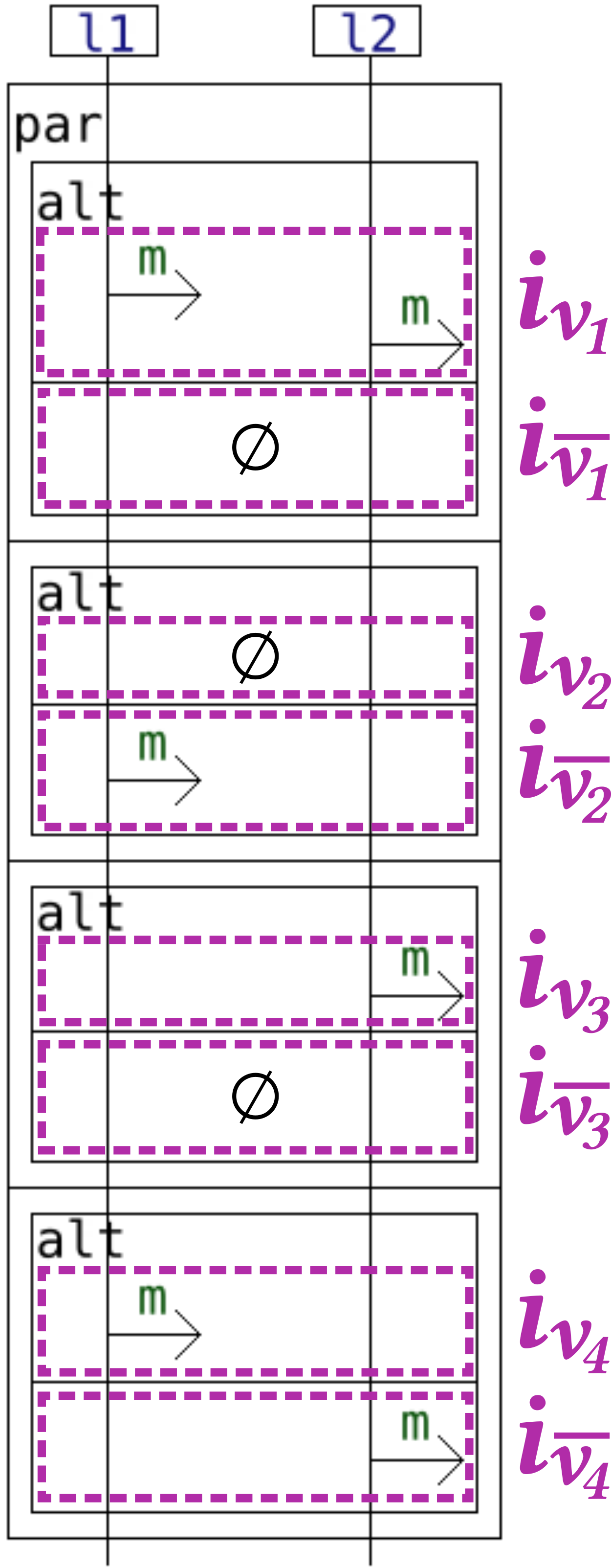}
\caption{Reduction}
\label{fig:ex_1in3sat_simple}
\vspace{-.75cm}
\end{wrapfigure}

The same reasoning can be applied as for the relationship between clause $C_2$ and component $\sigma_2$.

In other words, during the execution of $i$, given the use of exclusive alternative operators in $alt(i_v,i_{\bar{v}})$ sub-terms, the choice of either one of the $alt$ branch constitutes an assignment of Boolean variable $v$. The overall parallel composition then simulates all possible variable assignments (i.e. the search space for $\rho$).

Then, the satisfaction of $\phi$ as the conjunction of clauses $C_1$ and $C_2$ in 1-in-3-SAT is equivalent to that of $\mu=(\sigma_1,\sigma_2) \in AccMult(i)$. Indeed, the same $\rho$ must be used to solve both $C_1$ and $C_2$ and the same global execution of $i$ must be used to consume both $\sigma_1$ and $\sigma_2$ exactly.

In our example, $\phi = (v_1 \vee \overline{v_2} \vee v_4) \wedge (v_1\vee v_3 \vee \overline{v_4})$ is solvable in 1-in-3-SAT by $\rho:[v_1 \rightarrow \bot, v_2 \rightarrow \top, v_3 \rightarrow \top, v_4 \rightarrow \top]$. This is equivalent to the fact that $\mu=(l_1!m,l_2!m)$ is consumed exactly by the execution of $i$ from Fig.\ref{fig:ex_1in3sat_simple} when $i_{\overline{v_1}}$ is chosen over $i_{v_1}$,  $i_{v_2}$ over $i_{\overline{v_2}}$, $i_{v_3}$ over $i_{\overline{v_3}}$, and $i_{v_4}$ over $i_{\overline{v_4}}$. For any such 3-CNF formula $\phi=C_1\wedge C_2$ defined over $V=\{v_1,\cdots,v_4\}$, the 1-in-3-SAT problem can therefore be reduced to that of the membership of $(l_1!m,l_2!m)$ w.r.t. the interaction $i$ constructed from $\phi$ as above.

As explained earlier, this sketch of proof can be extended to include any numbers $p$ and $q$ of resp. variables and clauses. It suffices to consider $q$ lifelines $l_1,\cdots,l_q$, the multi-trace $\mu=(l_1!m,\cdots,l_q!m)$ and $p$ parallelized sub-interactions $alt(i_{v_1},i_{\overline{v_1}})$, $\cdots$,  $alt(i_{v_p},i_{\overline{v_p}})$. This generalized version is detailed in the annex.

Given that we have identified a case of multi-trace membership equivalent to a NP-complete problem, by reduction, multi-trace membership is NP-hard.

\subsection{Toward testing multi-traces against interactions}

When testing DSs, the $\omega$ function cannot be directly used. Indeed, the absence of global clock in a DS might cause difficulties to synchronize the cessation of observations on the different sub-systems. Let us consider for example the term $i=strict(a!m,b?m)$ and a system under test composed of a sub-system $Sys_a$ (implementing lifeline $a$) communicating with a system $Sys_b$ (implementing lifeline $b$). A tester connected to $Sys_a$ might observe the empty execution $\epsilon$ because he stopped logging before the occurrence of emission $a!m$. On the other hand a tester connected to $Sys_b$ could have logged long enough to observe $b?m$ so that the overall observation would correspond to $\mu=(\epsilon, b?m)$. As $\mu$ is a strict prefix of an accepted multi-trace, but not an accepted multi-trace itself, $\omega(i,\mu)$ would return $Fail$. From a testing perspective one would like to conclude that, even though $\mu$ is not an accepted multi-trace, a longer observation on $Sys_a$ might have enabled the observation of $a!m$, yielding to the global observation $(a!m,b?m)$. As such, $\mu$ does not reveal a fault, contrary to observations that are not prefixes of accepted multi-traces.

In all generality, being able to determine whether or not a multi-trace is a prefix of an accepted multi-trace would require techniques that are beyond the scope of this paper. However, we propose a first approach in which we simply adapt $\omega$ to identify couples $(i,\mu)$ for which we can not make a decision (and therefore provide them with dedicated verdict of inconclusiveness).

Rules of Def.\ref{def:ana_rules} are adapted as follows. For a couple $(i,\mu)$, if $\mu$ is empty and $i$ does not express $\epsilon$, rule $R2$ applies and leads to $UnCov$. However, we could rather return a $TooShort$ verdict (as we did for trace analysis in \cite{fase2020}). Intuitively the existence of an execution (from an initial couple $(i_0,\mu_0)$) leading to this verdict will prove that the trace leading to it is a prefix of a trace accepted by $i_0$, but not a trace accepted by $i_0$. Now, if $\mu$ is not empty and from $(i,\mu)$ no first element of a component of $\mu$ matches a frontier action of $i$, rule $R4$ applies and leads to the $UnCov$ verdict. However we can here distinguish between cases in which an observability problem (as discussed earlier) may arise and cases in which it does not. 
(a) If no component of $\mu$ has been emptied, this means that, at this point of the test, observations on all sub-systems are still ongoing. We are therefore sure of having a true error because no further execution complies with the interaction model. We may therefore return an $Out$ verdict.
(b) If at least one component of $\mu$ has been emptied, this means that the tester ceased logging on the corresponding sub-system. It might be that a longer observation on this sub-system would have enabled the application of rule $R3$. However we lack this information and hence return a $LackObs$ verdict.

Those considerations are reflected in the rules by slightly changing $R1$ (replacing $UnCov$ by $TooShort$) and by sub-dividing $R4$ into rules $R4a$ and $R4b$ which discriminate between the 2 aforementioned cases and are s.t.:

\begin{minipage}[t]{.9\textwidth}
\begin{prooftree}
\AxiomC{$i$}
\AxiomC{$(\sigma_1,\cdots,\sigma_n)$}
\LeftLabel{$(R4a)$}
\RightLabel{$
\left\{
\begin{array}{l}
\forall j \in [1,n],~ \sigma_j \neq \epsilon,\\
\forall j \in [1,n],~ \forall p\in  front(i),~ fst(\sigma _j)\neq i_{|p} 
\end{array}
\right.
$}
\BinaryInfC{$Out$}
\end{prooftree}
\end{minipage}

\vspace*{.75cm}

\begin{minipage}[t]{.7\textwidth}
\begin{prooftree}
\AxiomC{$i$}
\AxiomC{$(\sigma_1,\cdots,\sigma_n)$}
\LeftLabel{$(R4b)$}
\RightLabel{$
\left\{
\begin{array}{l}
(\sigma_1,\cdots,\sigma_n) \neq (\epsilon,\cdots,\epsilon)\\
\land (\exists j \in [1,n] \text{ s.t. } \sigma_j=\epsilon)\\
\wedge 
\left(
\begin{array}{l}
\forall j \in [1,n],~ \forall p\in front(i),\\
\sigma _j \neq \epsilon \Rightarrow fst(\sigma_j)\neq i_{|p}
\end{array}
\right)
\end{array}
\right.
$}
\BinaryInfC{$LackObs$}
\end{prooftree}
\end{minipage}

\vspace*{.75cm}

We can then formulate a new $\widetilde{\omega} : Int \times Mult \rightarrow \{Pass,WeakPass,Inconc,Fail\}$ function s.t. for any $i \in Int$ and $\mu \in Mult$ we have:

\begin{itemize}
    \item $\widetilde{\omega}(i,\mu) = Pass$ iff there exists a path $(i,\mu)$ \resizebox{!}{10pt}{$\overset{*}{\leadsto}$} $Cov$
    \item $\widetilde{\omega}(i,\mu) = WeakPass$ iff there does not exist a path $(i,\mu)$ \resizebox{!}{10pt}{$\overset{*}{\leadsto}$}
            $Cov$ and there exists a path $(i,\mu)$ \resizebox{!}{10pt}{$\overset{*}{\leadsto}$} $TooShort$
    \item $\widetilde{\omega}(i,\mu) = Inconc$ iff there does not exist a path of the form $(i,\mu)$ \resizebox{!}{10pt}{$\overset{*}{\leadsto}$} $v$ with $v \in \{Cov,TooShort\}$ 
    and there exists a path $(i,\mu)$ \resizebox{!}{10pt}{$\overset{*}{\leadsto}$} $LackObs$
    \item $\widetilde{\omega}(i,\mu) = Fail$ if all paths lead to $Out$
\end{itemize}

$\widetilde{\omega}$ is just a first proposal to deal with DS testing, and we will, in future works propose more precise methods for applying multi-trace analysis to testing.
For example, instead of producing $LackObs$ when the conditions of applications of $R4b$ holds,  we could explore the possible future characterized by $i$, pursuing the goal of identifying continuations starting by actions of $Act(l_j)$, for some $j\leq n$ such that $\sigma _j=\epsilon$. We could then try to produce more precise verdicts, by reasoning on the existence of such continuations, obtained by simulating some executions of $i$ rather than by consuming actions in $(\sigma_1,\cdots ,\sigma_n)$.

\section{Conclusion\label{sec:conclusion}}

We have proposed an approach to decide on the membership of multi-traces w.r.t. semantics defined on interaction models. The analysis consists in applying non-deterministic reading of the multi-trace using small-steps of the operational semantics. This approach have been validated with formal proofs of correctness using Coq, and a study on complexity. Moreover, a prototype tool that implements this analysis method has been developed in line with theoretical claims. Finally, we have discussed how membership analysis can be extended for testing distributed systems where logging of multi-traces is performed under observability limitations.
This last subject, with that of introducing data in models will be the objects of further works.

\bibliographystyle{splncs04}
\bibliography{biblio}

\clearpage

\appendix
\section{NP-hardness sketch of proof in general case}

For practicity, we will use $n-ary$ notations for binary operators $f \in \{strict,$ $seq,par,alt\}$, with $i=f(i_1,\cdots,i_n)$ designating the folding of $f$ s.t. $i=f(i_1,f($ $\cdots,f(i_{n-1},i_n)\cdots))$. We now reduce the 1-in-3-SAT boolean satisfiability problem to multitrace membership so as to prove the NP-hardness of the latter.

Let us consider a 3-CNF formula $\phi = C_1\wedge \cdots \wedge C_q$ defined over a set $V = \{v_1,\cdots,v_p\}$ of Boolean variables. $\phi$ being a 3-CNF formula, for any $j\in [1,q]$, $C_j$ is a disjunction of 3 literals $\alpha_j \vee \beta_j \vee \gamma_j$ (i.e. $\alpha_j$, $\beta_j$ and $\gamma_j$ are of the form $v$ or $\overline{v}$ with $v \in V$ and $\bar{~}$ the negation operator).

The 1-in-3 SAT problem consists in finding a solution $\rho : V \rightarrow \{\top,\bot\}$ s.t. for every clause $C_j$ only one of $\alpha_j$, $\beta_j$ or $\gamma_j$ is set to $\top$.

For any $k\in [1,p]$ let us define $i_k=alt(i_{v_k},i_{\overline{v_k}})$ with $i_{v_k} = seq( i_{v_k}^1, \cdots, i_{v_k}^q )$ and $i_{\overline{v_k}} = seq( i_{\overline{v_k}}^1, \cdots, i_{\overline{v_k}}^q )$ s.t. for all $j \in [1,q]$, given clause $C_j$ we have:\\
\noindent $\bullet$ if $v_k$ occurs in $C_j$ then $i_{v_k}^j = l_j!m$ and else $i_{v_k}^j = \varnothing$\\
\noindent $\bullet$ if $\overline{v_k}$ occurs in $C_j$ then $i_{\overline{v_k}}^j = l_j!m$ and else $i_{\overline{v_k}}^j = \varnothing$

Let us then consider $i=par(i_1,\cdots,i_p)$ as illustrated on Fig.\ref{fig:ev_1in3sat}. For instance, given $C_1 = v_1 \vee v_2 \vee \overline{v_p}$ we have $i_{v_1}^1 = l_1!m$, $i_{\overline{v_1}}^1 = \varnothing$, $i_{v_2}^1 = l_1!m$, $i_{\overline{v_2}}^1 = \varnothing$, $i_{\overline{v_p}}^1 = l_1!m$ and $i_{v_p}^1 = \varnothing$. Likewise, the other emissions of $m$ drawn in Fig.\ref{fig:ev_1in3sat} correspond to $C_2 = \overline{v_1} \vee v_2 \vee v_p$ and $C_q = v_1 \vee \overline{v_2} \vee \overline{v_p}$.

The 1-in-3-SAT problem $\phi$ is then equivalent to the multi-trace membership problem $\mu=(l_1!m, \cdots, l_q!m) \in AccMult(i)$.

\begin{wrapfigure}{l}{0.39\textwidth}
\vspace{-.75cm}
\centering
\includegraphics[width=.4\textwidth]{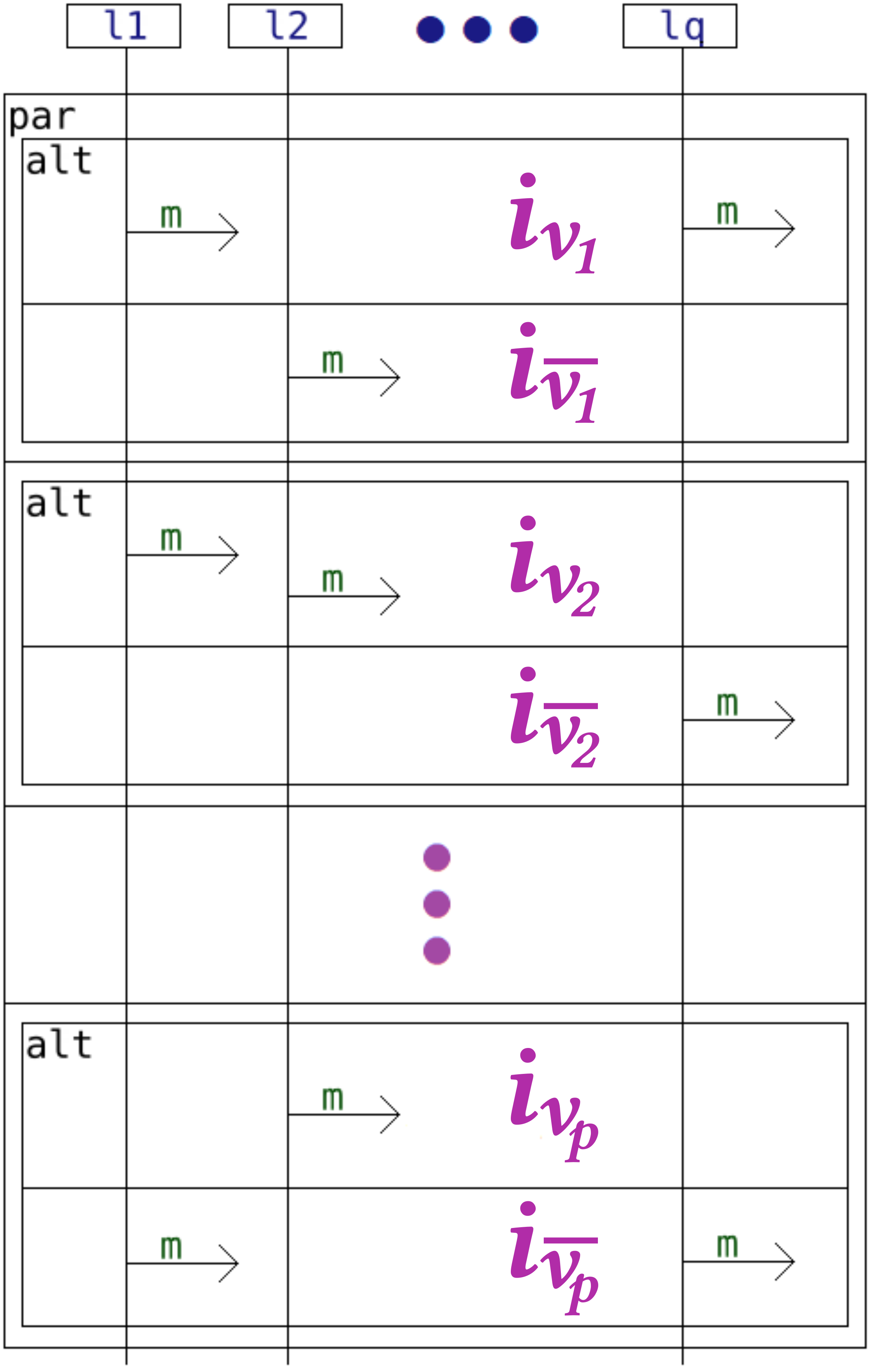}
\caption{$i$ obtained from $\phi$}
\label{fig:ev_1in3sat}
\vspace{-.75cm}
\end{wrapfigure}

Indeed, for any component $\sigma_j = l_j!m$ of $\mu$, $\sigma_j$ is expressed exactly once iff exactly one of the sub-interactions $i_{\alpha_j}$, $i_{\beta_j}$ or $i_{\gamma_j}$ is chosen during the execution of $i$ (choice w.r.t. their respective parent $alt$ operator). The satisfaction of component $\sigma_j$ is therefore equivalent to that of clause $C_j$. For instance, on the example from Fig.\ref{fig:ev_1in3sat}, $\sigma_1=l_1!m$ is satisfied iff only one of $i_{v_1}$, $i_{v_2}$ or $i_{\overline{v_p}}$ is chosen on their respective alternative branches, which exactly corresponds to the satisfaction of $C_1 = v_1 \vee v_2 \vee \overline{v_p}$ in 1-in-3-SAT,

In other words, during the execution of $i$, given the use of exclusive alternative operators in $alt(i_x,i_{\overline{x}})$ sub-terms, the choice of either one of the $alt$ branch constitutes an assignment of boolean variable $x$. The overall parallel composition then simulates all possible variable assignment (i.e. the search space for $\rho$).

Then, the satisfaction of $\phi$ as the conjunction of clauses $C_j$ in 1-in-3-SAT is equivalent to that of $\mu \in AccMult(i)$, given that the same execution of the model $i$ must satisfy conjointly every component $\sigma_j$.

\end{document}